\documentclass[lettersize,journal]{IEEEtran}

\usepackage{amsmath,amsfonts}
\usepackage{amssymb}  
\usepackage{amsthm}
\usepackage{amsmath}
\usepackage{array}
\usepackage{textcomp}
\usepackage{stfloats}
\usepackage{url}
\usepackage{verbatim}
\usepackage{graphicx}
\usepackage{cite}
\usepackage{subcaption}
\usepackage{cuted}
\usepackage[ruled,linesnumbered]{algorithm2e}
\usepackage[utf8]{inputenc}
\usepackage{setspace}
\usepackage{authblk}
\usepackage{tabularx}
\usepackage{tablefootnote} 
\usepackage[bottom]{footmisc}
\usepackage{soul,color}
\usepackage{xurl}

\usepackage{bibentry}

\usepackage[nolist]{acronym}

\begin{acronym}
\acro{LBS}{location based services}
\acro{FD}{federated distillation}
\acro{FL}{federated learning}
\acro{RSSI}{radio signal strength indicator}
\acro{IoT}{Internet of Things}
\acro{ML}{machine learning}
\acro{GPS}{global positioning system}
\acro{LOS}{line of sight}
\acro{NLOS}{non line of sight}
\acro{GNSS}{global navigation satellites system}
\acro{WSN}{wireless sensor network}
\acro{CSI}{channel state information}
\acro{UWB}{ultra-wide bandwidth}
\acro{MLP}{multi-layer perceptron}
\acro{ToA}{time of arrival}
\acro{TDoA}{time difference of arrival}
\acro{AoA}{angle of arrival}
\acro{RFID}{radio frequency identification}
\acro{MSE}{mean squared error}
\acro{MAE}{mean absolute error}
\acro{AP}{access point}
\acro{RP}{reference position}
\acro{CD}{co-distillation}
\acro{KD}{knowledge distillation}
\acro{NTK}{neural tangent kernel}
\acro{DNN}{deep neural network}
\acro{IPS}{indoor positioning paper}
\acro{I.I.D}{independently and identically distributed}
\acro{N.I.I.D}{non independently and identically distributed}
\acro{GAN}{generative adversarial network}
\acro{RMSE}{root mean squared error}
\acro{SoA}{state-of-the-art}
\acro{XR}{extended reality}
\acro{LoRaWAN}{long range wide area network}
\acro{MLP}{multi-layer perceptron}
\end{acronym}

\usepackage{cleveref}
\usepackage{balance}
\usepackage{multicol}

\newtheorem{remark}{Remark}
\newtheorem{theorem}{Theorem}
\newtheorem{lemma}[theorem]{Lemma}

\begin{document}
\bstctlcite{IEEEexample:BSTcontrol}

\title{Federated Distillation based Indoor Localization for IoT Networks}





\author[1]{Yaya Etiabi}
\author[2,3]{Marwa Chafii}
\author[1]{El Mehdi Amhoud}
\affil[1]{School of Computer Science, Mohammed VI Polytechnic University, Benguerir, Morocco}
\affil[2]{Engineering Division, New York University (NYU) Abu Dhabi, UAE }
\affil[3]{NYU WIRELESS, NYU Tandon School of Engineering, Brooklyn, NY
\authorcr Email: {  \{yaya.etiabi, elmehdi.amhoud\}@um6p.ma, marwa.chafii@nyu.edu} \vspace{-2ex}} 
\maketitle
\begin{abstract}
Federated distillation (FD) paradigm has been recently proposed as a promising alternative to federated learning (FL), especially in wireless sensor networks with limited communication resources. However, all state-of-the-art FD algorithms are designed for classification tasks only and less attention has been given to regression tasks. In this work, we propose an FD framework that properly operates on regression learning problems. Afterwards, we present a use-case implementation by proposing an indoor localization system that shows a good trade-off communication load vs. accuracy compared to FL-based indoor localization. With our proposed framework, we reduce the number of transmitted bits by up to 98\%. We analyze the energy efficiency regarding the number of communication rounds it takes for both FD and FL systems to achieve the same localization accuracy. The results revealed that FD comes with greater energy efficiency by reasonably saving transmission energy at the expense of computation energy. This is a substantial advantage for battery-powered IoT systems with very limited bandwidth. Moreover, we show that the proposed framework is much more scalable than FL, thus more likely to cope with the
expansion of wireless networks.
\end{abstract}


\section{Introduction}
\label{sec:introduction}
\IEEEPARstart{L}{ocation} based services
play an important role in several applications by providing targeted information to individuals or entities based on their geographic location in real or near-real time, typically through wireless communication networks. These applications include navigation, individual tracking, emergency services, asset tracking, logistics planning, workforce management, location-based advertising and social networking. 

In recent decades, \ac{LBS} have grown considerably and are now more than ever at the core of the digital revolution we are witnessing. Their market is expected to reach $\$318.64$ billion in $2030$ \cite{web} due to the increase of demand from different sectors such as agriculture, defense, transportation, energy, healthcare, etc. This growth is the direct result of the evolvement of the underlying technologies centered around wireless sensor networks. 

Today with the advent of 5G \cite{li2018}, and the prospect of 6G \cite{jiang2021}, the number of connected devices will grow at an unprecedented rate, resulting in the massive deployment of \ac{IoT} infrastructures \cite{Jouhari2022ASO}. Moreover, 6G comes with new applications such as multisensory \ac{XR} applications, connected robotics and autonomous systems \cite{saad2020}. It is also expected to offer sensing and localization as new services \cite{carlos2021}. All these transformational applications drive the need for accurate localization systems which require lots of resources due to the massive deployment of \ac{IoT} devices. 
Yet in the literature, in addition to the range-free techniques such as centroid method \cite{centroid} and distance vector hop (DV-Hop) technique \cite{dvhop}, typical ranging techniques based on \ac{CSI}\cite{csi} , \ac{AoA} \cite{aoa}, \ac{ToA} \cite{toa}, \ac{TDoA} \cite{tdoa}, and \ac{RSSI} \cite{rssi} using various wireless technologies such as \ac{RFID} \cite{rfid}, \ac{UWB} \cite{uwbtoa}, WiFi \cite{wifi}, LoRaWAN \cite{Etiabi2022SpreadingFA}, and Bluetooth \cite{bluetooth}
have been proposed for indoor positioning. All of these techniques present a number of issues, including low precision, high computational complexity, and unreliability due to wireless channel impairments such as multipath effects caused by \ac{NLOS} propagation in indoor environments, while most positioning devices lack sufficient computing power. Moreover, these techniques require the building of empirical models that capture all the channel effects including noises, multipath fading, channel variations due to the indoor movements, leading to a very high time complexity and  a high dependence to channel conditions.

Consequently, due to the difficulty to derive robust models that capture these indoor channel impairments, researchers turn towards data-based localization using
\ac{ML} which does not require empirical models but rather uses constructed datasets that capture all variations in the indoor environments. Indeed, \ac{ML} is used to tackle the aforementioned limitations of traditional methods
as enlightened in the works in \cite{MLSurvey, uwb, rssilora, csi1,csi2,csi3}.
In these works, we can see that \ac{ML} is a very promising  and game-changing technology for \ac{IoT} localization in the sense that it provides robust and scalable localization systems with  improved accuracy and relatively low complexity\cite{scalableML}.

Nonetheless, ML based localization systems require important data exchange between IoT devices and the central server. To proceed the later, ML operation with a respect of users privacy, \acl{FL}, has been introduced.

 \Acf{FL} has been firstly introduced by Google researchers in \cite{FL}. It is a \ac{ML} paradigm that enables end-devices to collaboratively learn a shared prediction model while keeping all the training data on-device. Such approach decouples the ability to do \ac{ML} from the need to store the data in the cloud. Research works in \cite{tran2019, dinh2021, zhou2021, bennis2021} show the prominence of using such algorithm. Indeed, in \cite{FedLoc2020}, authors have proposed  a \ac{FL}  based localization framework coined \textit{FedLoc}, which richly addresses the privacy concern in collaboration among numerous IoT devices.
 
 However, in an \ac{IoT} network where a massive number of sensors is deployed, \ac{FL} is not enough since when implemented over capacity-constrained  communication  links,  communication  cost  and  latency  may severely  limit  the  performance  of  \ac{FL}. Therefore, \ac{FD} has been proposed in \cite{dist5_2018} to deal with these issues. In fact, \ac{FD} is a compelling distributed learning solution that only exchanges the model outputs whose dimensions are commonly much smaller than the model size. In the literature, \ac{FD} has been used in several applications implying only classification tasks such as handwritten digits recognition \cite{mnist} and image classification \cite{cifar}.
 In this work, we consider a FD based regression problem that predicts the location of an IoT device given its RSSI measurements in that location.
 Unlike the \acl{SoA} FD approaches, our localization problem is a regression one, hence the need to develop a \ac{FD} framework for regression problems.
 
 The motivation of this paper lays in three points: (i) Secure crowdsourcing and decoupling the need to store the RSSI fingerprints data in a fusion center from the ability to train the \ac{DNN} model. Indeed, RSSI fingerprints  are crowdsourced by a pool of IoT devices all over the network. Instead of sending the crowd-sourced data to the fusion center, each device keeps its data locally for privacy preservation and leverage it to train a local model. The global model is then obtained through the federation of the different local models. This is the fundamental of \ac{FL};
 (ii) Reducing as much as possible the communication load involved in the federated training  of the localization model. In fact, the IoT devices namely FL clients are assumed to have limited bandwidth leading to a constrained model parameters exchange. Hence the introduction of FD for only output exchange which is much less bandwidth consuming.
Note that communication is more challenging for IoT devices. 
Therefore, reducing the number of transmitted bits may lead to greater energy efficiency of FD system for the same performance goal as FL;
 (iii) Constructing a FD framework for regression problems for a better implementation of the localization model which is a regression problem by design. In fact, \acl{SoA} FD algorithms do not deal with the regression problems.
 
The main contributions of this work are summarized as follows:
\begin{itemize}
    \item We develop a \ac{FD} framework for regression tasks since the previously proposed \ac{FD} algorithms are dealing with classification problems and no attention has been given to the regression ones. We validate our \ac{FD} framework through performance evaluation using different publicly available experimental datasets.
    \item We propose an IoT localization system based on our proposed \ac{FD} framework. We prove that this localization framework works for indoor systems as well as for outdoor systems. Also to the best of our knowledge, this work is the first to tackle localization problem under a \ac{FD} framework which considerably reduces the communication complexity over bandwidth-constrained wireless networks.
\end{itemize}

The reminder of this paper is organized as follows. Section~\ref{sec:relatedwork} presents related works. In Section~\ref{sec:proposedframework}, we describe our proposed \ac{FD} framework followed by its performance evaluation in Section~\ref{sec:performanceevaluation}.
Finally, Section~\ref{sec:conclusion} concludes our work and provides future research directions.

\section{Related Work}
\label{sec:relatedwork}

In this section, we present the works related to the following: \ac{RSSI} fingerprinting based localization, \ac{FL} based IoT localization frameworks and federated distillation.

\subsection{\ac{RSSI} fingerprinting based localization}
Being the key pillar of \ac{LBS}, localization in IoT networks is a field of study that has started very early with the first
\acp{WSN}. At its early stage, \ac{GNSS} such as the \ac{GPS} were widely used to perform sensor nodes localization with some appropriate algorithms \cite{GPS}. 
With the increase of sensor network density and the expansion of IoT devices’ deployment mainly in indoor environments, GPS solutions become too expensive and suffer from indoor environment’s impairments. To address this issue, many developed localization algorithms do not use GPS directly, but employ it as an assistance in some cases \cite{etiabi2020,etiabi2021} while more efforts are made on the excavation of the \ac{IoT} network itself by exploiting wireless signal characteristics including \ac{CSI}, \ac{AoA}, \ac{ToA}, \ac{TDoA}, \ac{RSSI}. The most commonly used signal, due to its lowest cost and high availability, is the \ac{RSSI} which captures the power-distance relationship of a signal propagation environment. Indeed, fingerprinting based methods  firstly construct a \ac{RSSI} database by extracting  the  representative and distinguishable parameters, namely fingerprint, from the received signals of IoT devices at different locations, then use the constructed database to predict the location corresponding to a new RSSI recording  based on its similarity with the recorded fingerprints.

However \ac{RSSI} measurements are challenged by channel impairments such as the multipath fading caused by the \ac{NLOS} signals leading to high complexity of analytical solutions. These problems are accentuated by the growth of IoT networks which involves a systematic adaptation for localization schemes in order to maintain the mainstream of \ac{LBS}.
Therefore, \ac{ML} has been introduced to maintain a good
trade-off between accuracy and complexity as shown in \cite{Njima2019, dnn_blindLoc,Gogolak2011IndoorFL} where many deep learning frameworks have been proposed for RSSI fingerprinting based localization.


\subsection{FL based IoT localization frameworks}
Although RSSI based fingerprinting has gathered researchers' attention in the past decade and has achieved a great breakthrough, it suffers from an inherent problem: The localization accuracy degrades  abruptly  over  time  due  to  the  very dynamic  environment and  unstable  wireless  devices leading to a high calibration effort for  fingerprint  collection.\\
Thus, by using conventional \ac{ML}, we end up with a huge fingerprint database in the server with sometimes privacy issues. Therefore, employing  \ac{FL}  allows  keeping users private data locally, thus ensuring the privacy preservation and reducing the dataset size while enabling an adequate localization performance.
Consequently, researchers are turning towards the design of localization systems based on federated learning models. Indeed, the work in \cite{FLoc2019}, is the first to propose a federated learning framework called \textit{FLoc}, an RSSI fingerprinting based  indoor  localization scenario to solve the security problem in fingerprint database updating localization. Similarly in \cite{PFL2020, FLLcrowd2020,FedLoc2020}, federated learning is driving the design of localization systems.
Yet, besides the great advantage of privacy preservation and training computation load reduction that FL brings in the cited works, it fails to meet  the bandwidth and energy constraints of the real world IoT systems since the communication load and energy consumption involved in the federated training is not often  suitable for IoT devices with very limited capabilities. Hence, in this work, we overcome these limitations of FL by utilizing federated distillation(FD). The first advantage of the FD lies indeed, in the flexibility it brings to devices to train tailored model architectures. This feature is beneficial for networks with heterogeneous constraints. Additionally, FD comes with augmented robustness in the sens that  adversarial or malicious clients can not directly influence the parameterization of the model during training. However, the most substantial benefits of FD resides in its great communication efficiency compared to FL.


\subsection{Federated distillation}
Federated distillation combines the ideas of knowledge distillation and federated learning. Knowledge distillation \cite{kdist} is a student-teacher approach in which, a student model is helped in its training by a pretrained teacher model in order to transfer the knowledge of the teacher to the student and thus speed up the learning process. It is worth noting that both the student and the teacher have access to the same dataset. A distributed version of this approach namely co-distillation has been proposed in \cite{cdist}, where instead of having a pretrained teacher,  each student considers the aggregated knowledge of the remaining students as its teacher, leading to an online teacher model training alongside students models.
FD can be seen as a derivative of co-distillation as shown in \cite{fdist} with the difference that each FD student possesses its own data as in FL. It is proved therein through asymptotical analysis leveraging the kernel method in \cite{ntk} that the FD can achieve even better performance than knowledge distillation specifically if in the latter, the teacher model is not well pretrained. Therefore, recent works \cite{dist5_2018, dist2_2020,dist3_2020,dist4_2020,dist1_2021} introduces FD to tackle the limitations of FL in terms of bandwidth requirements and energy constraints.

However, to date, all the proposed FD frameworks are targeting only classification tasks and no attention has been given to regression tasks.
Since several real-world ML target variables are continuous and need to be treated as regression problems, we first propose a federated distillation framework for regression. Then, we leverage it to build an indoor localization system.

\section{Proposed Framework}
\label{sec:proposedframework}

\subsection{Localization system model}
In traditional settings, RSSI data from multiple anchor nodes are collected through  a crowdsourcing system and merged in a server  for centralized training of the \ac{DNN} based localization model. In this work, we substitute this by a decentralized approach where collected RSSI data are no longer sent to the server. 
Instead, a pool of IoT devices, namely FD clients is selected to crowdsource RSSI data all over the network. Then, each selected device train a local model with its crowdsourced data.
In FD training setting, FD clients share their knowledge by exchanging their models' outputs so that at the end of the training, the local models predict nearly the same positions given new RSSI measurements.
Thus, for a node to be localized in this network, it records RSSI from in-range anchors and sends them to the nearest FD client which outputs the corresponding position by running its model. Note that the predicted position remains essentially the same as if the prediction was done by any other FD client. Indeed, the models are trained in a manner that provides them a global knowledge of the overall network. 
The FD system architecture and the whole FD training process are described in the next sections.

\subsection{Federated distillation system architecture}
The overall framework contains
three main components namely the parameter server called coordinator also known as teacher, the clients or workers also known as students, and the communication system. 
\subsubsection{Server}
The server (coordinator or teacher) is the entity that supervises the whole \ac{FD} training process. It starts by defining all the common parameters such as the total number of communication rounds, the number of clients to be selected for each communication round, the wireless communication settings.
The server is also in charge of defining the  global \ac{DNN} model. The necessary information to reproduce locally the global model is then broadcasted to the clients by the help of the communication system.
\subsubsection{Clients} The clients (workers or students) are in charge of training their respective local models using their private datasets. During the training, workers periodically upload their local average estimations per each target segment to the parameter server which will aggregate them to obtain a global estimation per each corresponding segment. The global estimations from the server are then used by the clients to update their respective loss functions for the next local training phase as shown in Fig.~\ref{fig:fdr}.
\subsubsection{Communications system} The role of the communication system as its name suggests is to deal with the wireless communication techniques used to communicate the model and the target estimations from the server to clients and vice-versa. This module is further discussed in Section~\ref{sec:com}.
\begin{figure}[!t]
    \centering
    \begin{subfigure}{.5\textwidth}
     \centering
     \includegraphics[scale=0.4]{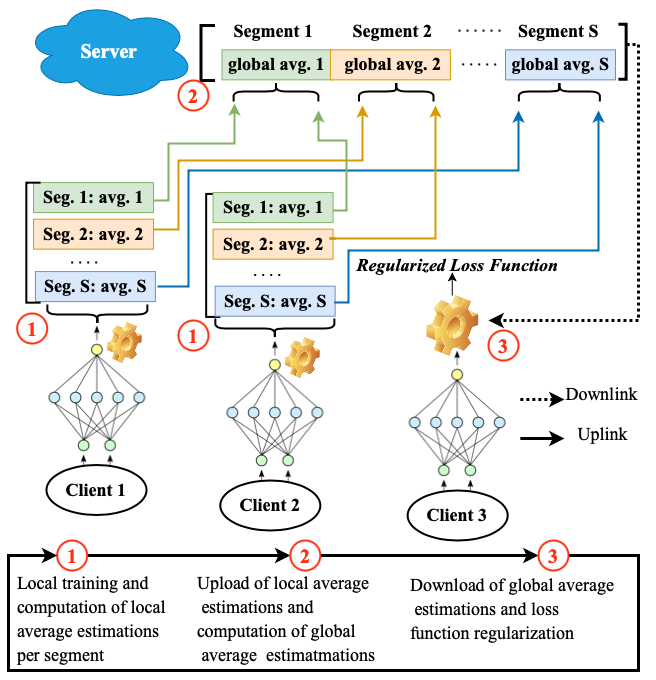}
    \caption{Illustration of segment-based knowledge sharing}
    \label{fig:fdr}
    \end{subfigure}
    \par\bigskip
     \begin{subfigure}{.5\textwidth}
     \centering
     \includegraphics[scale=0.3]{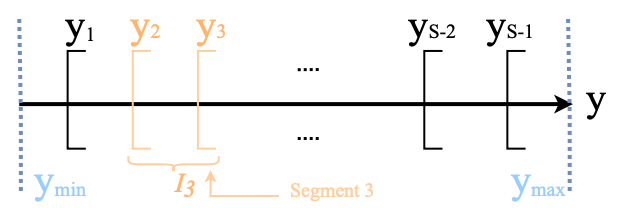}
    \caption{Target variable segmentation}
    \label{fig:seg}
    \end{subfigure}
    \caption{Illustration of FD based Regression training process.}
    \label{fig:fdreg}
\end{figure}
\begin{figure}[!t]
    \centering
    \includegraphics[scale=0.4]{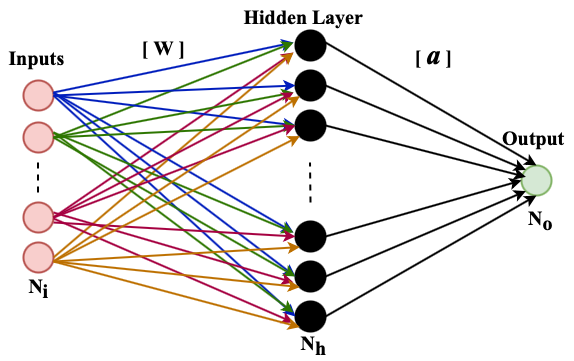}
    \caption{3-Layer Multi-Layer Perceptron Architecture.}
    \label{fig:dnn}
\end{figure}

\subsection{Problem formulation}
The objective of our proposed framework is to cooperatively train a \ac{DNN} using a communication-efficient federated learning scheme namely federated distillation, in a network of IoT devices. 
For the sake of convergence analysis in section~\ref{sec:conv}, we consider a 3-layer neural network comprising an input layer, a hidden layer and output layer with respectively $N_i$, $N_h$, and $N_o$ neurons
as shown in Fig. \ref{fig:dnn}. 
Given an input vector $X_i$, the prediction of the target $y_i $ is given by $\hat y_i = F_{\Theta}(X_i)$, where $F_{\Theta}$ is the function representing our \ac{DNN} and $\Theta = [W\; \alpha]$(see Fig. \ref{fig:dnn}) is the set of the \ac{DNN}'s weights.
Then, for a standalone training, the goal is to minimize the loss function given by: \begin{equation}
    \mathcal{L}(y,\hat y) = \sum_i\mathcal{L}(y_i,\hat y_i),
\end{equation} where 
\begin{equation}
    \hat y_i = F_{\Theta}(X_i) = \frac{1}{\sqrt{N_h}} \sum_{n=1}^{N_h}\alpha_n \ell_n(X_i),
\end{equation} 
with $\ell_n(X_i)= \sigma(W_n^{T} X_i)$ is the logit representing the output of the $n^{th}$ neuron of the last hidden layer, $\sigma(\cdot)$ a non-linear activation function, and $W_n$ the weights vector of the $n^{th}$ neuron. For the localization framework, the input $X_i$ is a vector containing the RSSI measurements from all the access points in the network and the target variable $y_i $ is the 2-D coordinates of the location where these RSSIs have been recorded.
To simplify the notations we will consider $F_{\Theta}(X_i) = F(X_i)$ in the next sections.

In vanilla distillation process so called knowledge distillation, each student independently learns the knowledge of the teacher by adding a regularization term regarding the gap between its prediction and the teacher's. 
Ideally, similar to the distillation in classification context in \cite{cdist}, the distillation for a learning task can be formulated as:
\begin{equation}
    \mathcal{L}(y,\hat y) = \underbrace{\sum_{i} \mathcal{L}\left(y_{i},\hat{y}_{i}\right)}_{\text {loss }}+\lambda \underbrace{\sum_{i}\sum_{n=1}^{N_h}\mathcal{L}\left( \ell_{n}(X_i),L_{n}(X_i)\right) }_{\text {distillation regularizer }},
    \label{equ:kd}
\end{equation}
where $L_{n}(X_i)$ is the pretrained teacher $n^{th}$ output logit given an input $X_i$ and $\lambda$ is the regularization coefficient, assuming that the teacher and the student have access to the same dataset.

On the other hand, \ac{CD} is a distributed version of this system where workers have to distill knowledge from each other. In fact, the main principle of \ac{CD} is to consider a set of predictions from several individual models as the teacher's knowledge, which is frequently more correct than the individual predictions. In this configuration, each student sees the ensemble of remaining students as its teacher.
Consequently, \eqref{equ:kd} is transformed into:
\begin{equation}
    \mathcal{L}(y,\hat y) =  \sum_{k}\sum_{i} \left[ \mathcal{L}\left(y_{i},\hat{y}_{i}\right) +  \lambda \sum_{n}  \mathcal{L}\left(L_{n}^k(X_i),\ell_{n}^k(X_i)\right)\right],
    \label{equ:cd}
\end{equation}
where $L_{n}^k(X_i) = \frac{1}{K-1}\sum_{k\prime \neq k}\ell_{n}^{k\prime}(X_i), \forall k$ with $K$ the number of students. Note that this loss function encompasses all the individual losses for students since this distributed student-teacher configuration creates an interdependence in the distillation process.
However, in a federated learning setting, both KD and CD are not realistic since the data privacy and heterogeneity are the cornerstone of federated learning systems. In the considered IoT network, IoT devices are assumed to possess only their own datasets. 
Consequently, similarly to the technique used to train classification models with FD where the samples are grouped by labels at each device and the mean prediction of each label group is shared with the network, we develop a FD scheme for regression problems. Indeed, for regression problems, since we do not have classes as labels,
the output variable being continuous, we will end up exchanging a vector with equal size to the dataset, which is not only in contradiction with our objective of reducing the number of transmitted bits but also unpractical regarding the \acl{N.I.I.D} aspect of the local datasets with respect to the output variables.
Thus, in order to get closer to the technique used for FD in classification context, under the mild assumption on the boundaries of the target values, the target $y$ is divided into $\mathbf{\mathcal{S}}$ segments as shown in Fig. \ref{fig:seg} so that the average predictions per segment are sent to the server for knowledge sharing (Fig. \ref{fig:fdr}).
Each segment $s$ is defined by a half-opened interval $\mathcal{I}_s $ defined as follows:
\begin{equation}
    \mathcal{I}_s = \left\{\begin{matrix}
 \left]-\infty, y^1 \right [&  if& s=1\\  \\
\left[y^{s-1}, y^{s} \right [&  if& 1<s<\mathcal{S}\\ \\
 \left[y^{s-1}, +\infty \right [&  if& s=\mathcal{S}\\ \\
\end{matrix}\right.,
\end{equation}
$$
\text{where }y^{s} = y_{min} + s\times \epsilon \text{ with } \epsilon=\frac{y_{max}-y_{min}}{\mathcal{S}}
$$
\begin{remark}
The number of segments $\mathcal{S}$ is a hyper-parameter depending on the resolution $\epsilon$ set for the learning task, i.e, regarding a regression problem, the resolution $\epsilon$ can be the maximum tolerable error in estimation. Then $\mathcal{S}$ is given by $$\mathcal{S}= \frac{y_{max}-y_{min}}{\epsilon}.$$
\end{remark}

\begin{remark}
Since the data is distributed across devices, computing the min and the max requires the server to know the boundaries of each local dataset, which may violate the privacy of the federated system. As such, one approach to get this done securely is through a multiparty computation scheme between devices.
However, for the sake of computation and communication efficiency, the server can set them arbitrary or request devices to upload their local min and max using cryptographic tools with order preserving encryption to preserve devices' data privacy.
\end{remark}

For a client $k$ in the network, we define the following settings.
\begin{equation}
\left\{\begin{matrix}
\begin{aligned}
\mathcal{I}_s^k &= \{i\} \text{ such that } y_i\in \mathcal{I}_s \\ & \text{ and } (X_i,y_i)\in \mathcal{D}_k, \text{ with }\left | \mathcal{I}_s^k \right | =  N_s^k \\ 
\bar{\ell}_{n,s}^k &=  \frac{1}{N_s^k } \sum_{i\in \mathcal{I}_s^k} \ell_{n,s}^k(X_i)  \label{eqn} \\
L_{n,s}^k  &=  \frac{1}{K-1 } \sum_{k\prime \neq k} \bar \ell_{n,s}^{k\prime} 
\end{aligned}
\end{matrix}\right.,
\end{equation}
where $\mathcal{I}_s^k$ is the partition of $k^{th}$ client's dataset belonging to the $s^{th}$ segment.  $\ell^k_{n,s}(x)$ is the $n^{th}$ output logit of client $k$ when the input $x$ belongs to segment $s$, $\bar{\ell}_{n,s}^k$ and $L_{n,s}^k$ are respectively the local  and global average of the $n^{th}$ output logit over the samples in segment $s$.

The above formulations leads to the new global loss function of the federated system given by:
\begin{equation}
 \begin{aligned}
    \mathcal{L}(y,\hat y) =& \sum_{k} \sum_{s} \left (\sum_{i\in \mathcal{I}_s^k}\mathcal{L}\left(y_{i},\hat{y}_{i}\right) + \right. \\
    & \left.\lambda \sum_{i\in \mathcal{I}_s^k} \sum_{n} \mathcal{L} \left(L_{n,s}^k,\ell_{n,s}^k(X_i)\right) \right ).
 \end{aligned}
    \label{eq:loss}
\end{equation}

\subsection{FD for regression algorithm }
Considering only the computation aspect of our framework, 
the overall process of solving the problem in \eqref{eq:loss} is described by Algorithm \ref{alg:fdr} which comprises both the client side pseudo-code and the server's.

In fact, in this configuration, in contrast to knowledge distillation, we do not have a pretrained teacher model. So both the student and the teacher are learning during the process especially since the teacher knowledge for a given student is actually the aggregated knowledge of the remaining students. 

As a consequence, at the initial communication round, it is not possible to obtain the regularization term of the loss function since no prediction has been done yet. That is why, as enlightened in \textit{``line 1 - line 3''} of Algorithm \ref{alg:fdr}, we start by training the local models for a few steps.
Until then, the distillation process can start (\textit{``line 6 - line 36''}).

The overall process in Algorithm \ref{alg:fdr} can be summarized in the following steps:
\begin{enumerate}
    \item Each client trains its local model using its private dataset and stores locally the average estimations per segment.
    \item Each client periodically uploads its local average estimations per segment to the server.
    \item The server computes the global average estimations per segment by averaging all the local average estimations per segment sent by all the clients.
    \item Each client downloads the global average estimations per segment and updates its loss function according to \eqref{eq:loss}.
    \item Each client repeats all the above steps until convergence.
\end{enumerate}
\SetKwComment{Comment}{/* }{ */}
\RestyleAlgo{ruled}
\begin{algorithm}
\SetAlgoLined
  \caption{Federated Distillation for Regression}\label{alg:fdr}
  \SetKwInOut{Input}{Inputs}
  \SetKwInOut{Output}{Outputs}
  \Input{$\{\mathcal{D}_k\}$ \Comment*[r]{Client's Datasets}}
  \Input{$\mathcal{L}(\cdot)$ \Comment*[r]{Loss function}}
  
  \Output{$\{\mathcal{W}_k\}$ \Comment*[r]{Clients trained models parameters}}
  \vspace{0.5cm}
  \ForEach{Client $k \in \{1,2, \dots,K\}$}{
     Local training of parameters $\{\mathcal{W}_k\}$ using $\{\mathcal{D}_k\}$
    }
    $r \gets 1$ \;
    get $\{L^k_{s,r}\}$ \;
  \vspace{0.5cm}
  \While{$r < max\_com\_rounds$}{
  \vspace{0.3cm}
  \SetKwProg{ClientLocalTraining}{ClientLocalTraining}{}{}
  \ClientLocalTraining{$(\mathcal{D}_k, \mathcal{W}_k,\{L^k_{s,r}\})$}{
    
     \ForEach{Client $k \in \{1,2, \dots,K\}$}{ 
     \Comment{In parallel}
     \For{$e \gets 1$ \KwTo $\#Epochs$}{%
      \ForEach{batch $b \in \mathcal{D}_k$}{%
        \ForEach{$X_{i,b}, y_{i,b} \in b$}{%
         $\begin{aligned}
         \rightskip=-1in
             &\mathcal{W}_{k} \leftarrow \mathcal{W}_{k}-\eta\nabla \{\mathcal{L}\left(\hat y_{i,b}, y_{i,b}\right)+\\
             &\lambda \cdot \mathcal{L}\left(L_{s(y_{i,b}),r}^k,
              \ell_{s(y_{i,b}),r}^k(X_{i,b})\right)\} 
         \end{aligned}$\;
        $\begin{aligned}
             \bar \ell_{s(y_{i,b}),r}^k \leftarrow &\bar\ell_{s(y_{i,b}),r}^k +\\
             &\ell_{s(y_{i,b}),r}^k(X_{i,b})\;
        \end{aligned}$\;
         
         $N_{s(y_{i,b}),r}^k \leftarrow N_{s(y_{i,b}),r}^k+1$\;
  
    }
    }
    }
    \For{$s \gets 1$ \KwTo $S$}{
     $\bar\ell^k_{s,r} \leftarrow \bar\ell^k_{s,r} / N_{s,r}^k$\;
    }
    }
    \KwRet{$\{\bar\ell^k_{s,r}\}$}\;
  }
  \vspace{0.5cm}
  \SetKwProg{ServerKnowlgdeDistillation}{ServerknowledgeDistillation}{}{}
  \ServerKnowlgdeDistillation{$(\{\bar{\ell}_{s,r}^k\})$}{
    \For{$s \gets 1$ \KwTo $S$}{
     \ForEach{Client $k \in \{1,2, \dots,K\}$}{
     $ L^k_{s,r} \leftarrow  L^k_{s,r} + \bar\ell^k_{s,r}$\;
    }
    }
    \ForEach{Client $k \in \{1,2, \dots,K\}$}{
     \For{$s \gets 1$ \KwTo $S$}{
     $ L^k_{s,r+1} \leftarrow  \frac{K\cdot L^k_{s,r} - \bar\ell^k_{s,r}}{K-1}$\;
    }
    }
    \KwRet{$\{L^k_{s,r+1}\}$}\;
  }
  \vspace{0.5cm}
   $r \gets r+1$ \;
  }
\end{algorithm}
\subsection{Convergence analysis}
\label{sec:conv}
After applying our segmentation technique and considering the \textit{mean square error (MSE)} as the loss function, i.e $\mathcal{L}(y,\hat y) = \mathbf{MSE}(y,\hat y)$ in (\ref{eq:loss}), we globally  try to solve the following optimization problem:

\begin{equation}
  \begin{aligned}
    \min_{\mathcal{W}_1,\mathcal{W}_2,\dots,\mathcal{W}_K} &\sum_{k}\sum_{s} \left [ \sum_{i\in \mathcal{I}_s^k}\left(y_{i}-\hat{y}_{i}\right)^2 + \right. \\
&\left. \lambda  \sum_{i\in \mathcal{I}_s^k} \sum_{n}  \left(L_{n,s}^k-\ell_{n,s}^k(X_i)\right)^2\right ].
\end{aligned}
\label{eq:opt}
\end{equation}
The convergence of the latter optimization problem can be established through the asymptotical analysis of the two foundational building blocks of federated distillation namely knowledge distillation (KD) \cite{kdist} and Co-Distillation (CD)\cite{cdist} by exploiting the theory of neural tangent kernel (NTK)\cite{ntk}.

\subsubsection{Convergence analysis of Knowledge distillation}
The problem of KD in this configuration is cast into the minimization of the objective function defined as follows:
\begin{equation}
\begin{aligned}
   J(W_k) = &\sum_{s} \sum_{i\in \mathcal{I}_s^k}\left(y_{i}-\hat{y}_{i}\right)^2+ \\
            &\lambda \sum_{s}\sum_{i\in \mathcal{I}_s^k}\sum_{n=1}^{N_h}\left( \ell_{n}\left(X_i\right)-L_{n}\left(X_i\right)\right)^2,
\end{aligned}
\end{equation}
where $W_k$ contains the weights of the the model of a client $k$.
We assume that the activation function $\sigma(\cdot)$ in the hidden layers is Lipschitz continuous and differentiable. The first derivative $\sigma^{\prime}(\cdot)$ is also assumed to be Lipschitz continuous.
Using the gradient descent algorithm with an infinitesimal step-size to solve the problem in \eqref{equ:kd} in the NTK settings \cite{ntk}  results in the following dynamics of the weights:
\begin{equation}
    \frac{\mathrm{d} }{\mathrm{d} \tau}\mathbf W_k^n = \sum_s \mathbf\Phi_n^s\left [ \frac{\alpha_n}{\sqrt{N_h}}\left ( \mathbf y^s-\hat{\mathbf y^s }\right )+\lambda \left ( \mathbf L_n^s - \boldsymbol{\ell}_n^s \right ) \right ]
    \label{equ:wdyn}
\end{equation}
\begin{proof}
In the kernel regime\cite{ntk}, with an infinitesimal step size $\eta$, the gradient descent iterations depicted by $W_k^n\left ( \tau +1 \right ) = W_k^n\left ( \tau \right ) - \eta \nabla J\left ( W_k^n\left ( \tau  \right ) ,\mathcal{D}\right ), \forall n = 1,2,...,N_h$ can be translated into continuous time domain variation given by: 
\begin{align*}
\frac{\mathrm{d} }{\mathrm{d} \tau}\mathbf W_k^n &= -\frac{\mathrm{d} }{\mathrm{d} \mathbf W_k^n}\mathbf J \left ( W_k, \mathcal{D}_k \right ) \\ 
 &= -\frac{\mathrm{d} }{\mathrm{d} \mathbf W_k^n}   \sum_s \sum_{i\in \mathcal{I}_s^k} \left ( \mathbf y_i - \frac{1}{\sqrt{N_h}} \sum_n \alpha_n\sigma\left ( \left ( \mathbf W_k^n \right )^TX_i \right ) \right )^2  \\ 
 &+  \lambda   \sum_s \sum_{i\in \mathcal{I}_s^k} \sum_n \left ( \mathbf L_n(X_i) - \sigma\left ( \left ( \mathbf W_k^n \right )^TX_i \right )  \right )^2   \\ 
 &=  \sum_s 2 \sum_{i\in \mathcal{I}_s^k}  X_i \cdot\sigma^{\prime}\left ( \left ( \mathbf W_k^n \right )^TX_i \right )\left ( \frac{\alpha_n}{\sqrt{N_h}} \left ( \mathbf y_i-\hat{\mathbf y}_i \right ) \right )  \\ 
 &+  \lambda   \sum_s 2 \sum_{i\in \mathcal{I}_s^k} X_i \cdot\sigma^{\prime}\left ( \left ( \mathbf W_k^n \right )^TX_i \right ) \left (  \mathbf L_n(X_i) -  \ell_{n}\left(X_i\right) \right )
\end{align*}
Let $\mathbf \Phi_n^s$ be the line vector containing the values $X_i \cdot\sigma^{\prime}\left ( \left ( \mathbf W_k^n \right )^TX_i \right )$, $\mathbf y^s$ the vector of labels $y_i$ such that $i\in \mathcal{I}_s^k$.
Similarly, $\boldsymbol{\ell}_n^s$ and $\mathbf L_n^s$ are the vectors containing the $n^{th}$ predicted logits of all the samples with labels in segment $s$ (i.e $y_i / i\in \mathcal{I}_s^k$) for the student and the teacher respectively.
Then, it follows that:
\begin{align*}
\frac{\mathrm{d} }{\mathrm{d} \tau}\mathbf W_k^n 
& =\sum_s \mathbf\Phi_n^s \left ( \frac{\alpha_n}{\sqrt{N_h}} \left ( \mathbf y^s-\mathbf{\hat{ y}}^s \right ) \right ) \\
& + \lambda  \sum_s \mathbf \Phi_n^s  \left (  \mathbf L_n^s -  \boldsymbol\ell_{n}^s \right ) \\
& =  \sum_s \mathbf\Phi_n^s\left [ \frac{\alpha_n}{\sqrt{N_h}}\left ( \mathbf y^s-\hat{\mathbf y^s }\right )+\lambda \left ( \mathbf L_n^s - \boldsymbol{\ell}_n^s \right ) \right ] 
\end{align*}
which is the results in \eqref{equ:wdyn}.
\end{proof}
Note that \eqref{equ:wdyn} reflects the modelling by a continuous-time differential equation of the convergence of a trajectory of the discrete algorithm to a smooth curve. Due to the difficulty to analyze the dynamics of the weights, \eqref{equ:wdyn} is translated in dynamics of the logits as:
\begin{equation}
    \frac{\mathrm{d} }{\mathrm{d} \tau}\boldsymbol{\ell}_k^n = \sum_s \left ( \mathbf\Phi_n^s \right )^T\mathbf\Phi_n^s\left [ \frac{\alpha_n}{\sqrt{N_h}}\left ( \mathbf y^s-\hat{\mathbf y^s }\right )+\lambda \left (\mathbf L_{n,k}^s - \boldsymbol{\ell}_{n,k}^s  \right ) \right ].
    \label{equ:ldyn}
\end{equation}
 Indeed,
 \begin{align*}
\frac{\mathrm{d} }{\mathrm{d} \tau}\boldsymbol{\ell}_k^n(X_i) 
& = \frac{\mathrm{d} }{\mathrm{d} \mathbf W_k^n}\boldsymbol{\ell}_k^n(X_i)  \cdot \frac{\mathrm{d} }{\mathrm{d} \tau}\mathbf W_k^n \\
& =  \frac{\mathrm{d} }{\mathrm{d} \mathbf W_k^n} \sigma\left ( \left ( \mathbf W_k^n \right )^TX_i \right ) \cdot \frac{\mathrm{d} }{\mathrm{d} \tau}\mathbf W_k^n \\
&= X_i\cdot \sigma^{\prime}\left ( \left ( \mathbf W_k^n \right )^TX_i \right )\cdot \frac{\mathrm{d} }{\mathrm{d} \tau}\mathbf W_k^n
\end{align*}

Therefore,
\begin{align*}
\frac{\mathrm{d} }{\mathrm{d} \tau}\boldsymbol{\ell}_k^n 
& = \sum_s \sum_{i\in \mathcal{I}_s^k}  \frac{\mathrm{d} }{\mathrm{d} \tau}\boldsymbol{\ell}_k^n(X_i) \\
& = \sum_s \sum_{i\in \mathcal{I}_s^k} X_i\cdot \sigma^{\prime}\left ( \left ( \mathbf W_k^n \right )^TX_i \right )\cdot \frac{\mathrm{d} }{\mathrm{d} \tau}\mathbf W_k^n \\
& = \sum_s \sum_{i\in \mathcal{I}_s^k} \left ( \Phi_n^s(X_i) \cdot \sum_{s^{\prime}} \mathbf\Phi_n^{s^{\prime}}\left [ \frac{\alpha_n}{\sqrt{N_h}}\left ( \mathbf y^{s^{\prime}}-\hat{\mathbf y^{s^{\prime}} }\right ) \right. \right.\\
 & \left. \left.+\lambda \left ( \mathbf L_n^{s^{\prime}} - \boldsymbol{\ell}_n^{s^{\prime}} \right ) \right ]\right )\\
& = \sum_s (\Phi_n^s)^T \Phi_n^s \left [ \frac{\alpha_n}{\sqrt{N_h}}\left ( \mathbf y^s-\hat{\mathbf y^s }\right )+\lambda \left ( \mathbf L_n^s - \boldsymbol{\ell}_n^s \right ) \right ].
\end{align*}
Let $\mathbf \Psi_n^s = (\mathbf \Phi_n^s)^T \mathbf \Phi_n^s$. Then, \eqref{equ:ldyn} becomes:
\begin{equation}
    \frac{\mathrm{d} }{\mathrm{d} \tau}\boldsymbol{\ell}_k^n = \sum_s \mathbf \Psi_n^s\left [ \frac{\alpha_n}{\sqrt{N_h}}\left ( \mathbf y^s-\hat{\mathbf y^s }\right )+\lambda \left (\mathbf L_{n,k}^s - \boldsymbol{\ell}_{n,k}^s  \right ) \right ]
\end{equation}

\begin{remark}
The matrix $\mathbf\Psi$ is referred to as a neural tangent kernel (NTK) \cite{ntk} whose theory establishes that the random initialization of the weights matrix and the over-parameterization hypothesis induce together a kernel regime such that $\mathbf\Psi(\tau)\sim \mathbf\Psi(0), \forall \tau>0$. This remark emanates from the fact that in over-parameterization setting, each weight vector along the trajectory of the gradient descent algorithm remains constant over time and maintains a very tight relationship to its initial value.
\end{remark}

 \begin{lemma}
 \label{lm:conv}
 The student network output converges asymptotically as follows:
 \begin{equation}
 \label{equ:lconv}
\lim _{\tau \rightarrow \infty} \mathbf{F^s}(\tau)= \mathbf{F}^s_{\infty} = \frac{1}{\alpha+\lambda}\left(\alpha \mathbf{y}^s+\lambda \sum_{n=1}^{N_{h}} \frac{\alpha_{n} \mathbf L_{n}^s}{\sqrt{N_{h}}}\right), \forall s
 \end{equation}
 where $ \alpha = \sum_{n=1}^{N_h} \frac{\alpha_n}{N_h}\alpha_n^2$.
 \end{lemma}
 
 \begin{proof}
 The theory of linear systems with a finite order \cite{Rahbar2020OnTU} has established that a linear dynamics of finite order can roughly represent the gradient descent behavior on the over-parameterized neural network, and the evolution of $\mathbf{F^s}(\tau)$ can be written as:
\begin{equation}
    \mathbf{F}^s(\tau) = \mathbf{F}^s_{\infty} + \sum_{j=1}^{d}\phi^s_je^{-z^s_j\tau}, \forall s
\end{equation}
with $d$ representing the order of the linear system. $\{\phi^s_j\}_{j=1}^d$ are complex-value vectors determined by the specifications of the dynamics, whereas  $\{z^s_j\}_{j=1}^d$ are the poles of the linear system.
Note that the non-zero complex-values $\{z^s_j\}_{j=1}^d$, \mbox{$ z^s_j\neq 0$} correspond the singular points of the Laplace transform of $\mathbf{F^s}(\tau)$ except the zero component $z^s_j = 0$ which corresponds to the constant $\mathbf{F}^s_{\infty}$.

Thanks to the assumptions on the eigenvalues of the matrices $\mathbf\Psi$ at initialization (i.e, under mild assumptions on eigenvalues of $\mathbf\Psi(0)$) in the kernel regime \cite{ntk}, \cite{Rahbar2020OnTU} has shown that with bounded inputs, bounded weights, all existing poles are positive-valued, yielding $\lim_{\tau \rightarrow \infty} \mathbf{F^s}(\tau) = \mathbf{F}^s_{\infty}$. Moreover, the analysis in infinite width regime of the neural network ($N_h\rightarrow \infty$) with the set assumptions demonstrates the calculation of  $\mathbf{F}^s_{\infty}$ and produces the result in \eqref{equ:lconv}.
 \end{proof}

 \begin{remark}
   The student's prediction error with respect to the teacher's can be estimated by:
\begin{equation}
\mathbf{e}^s = \left\|\mathbf{F}^s_{\infty}-\mathbf{y}^s\right\|_{2}=\frac{\lambda}{\alpha+\lambda}\left\|\mathbf{y}^s-\sum_{n} \frac{\alpha_{n} \mathbf L^s_{n}}{\sqrt{N_{h}}}\right\|_{2},  \forall s.
\end{equation}
This result implies that the learning accuracy of the student is highly correlated to the quality of the teacher in representing labels since for an imperfect teacher, the error $\mathbf e^s$ monotonically increases with $\lambda$.
 \end{remark}

\subsubsection{Federated distillation}
Federated distillation extends knowledge distillation to co-distillation where each worker possesses its own data and see the other workers as its teacher, leading to an online knowledge distillation.
Overall, the objective is the minimization of the global loss function defined in \eqref{eq:opt}.
Individually, each student minimizes its local loss function $J(\mathbf W_k)$ using knowledge distillation where the teacher knowledge is the ensemble prediction of the remaining students.
\begin{lemma}
Based on the results of \cref{lm:conv}, and the interactions between students, our optimization problem in ~\eqref{eq:opt} converges asymptotically as:
\begin{equation}
    \label{equ:cdconv}
    \lim _{r \rightarrow \infty} \mathbf{F}_k(r)=\mathbf{y}, \forall k \in\{1,2, \ldots, K\},
\end{equation}
where $r$ designates the $r^{th}$ communication round.
\end{lemma}

\begin{proof}
The first iteration of knowledge distillation is performed after initialization and a first run of the gradient descent by all students. Consequently, the first updates i.e, $\{\mathbf{F}^{k}(0)\}_{k=1}^K$ are shared. Then each student $k$ locally runs gradient descent and, by \cref{lm:conv}, its model converges to:
\begin{equation}
    \mathbf{F}^{k}(1)=\frac{1}{\alpha+\lambda}\left(\alpha \mathbf{y}+\frac{\lambda}{K-1} \sum_{k^{\prime}\neq k}^{K} \mathbf{F}^{k^{\prime}}(0)\right).
\end{equation}
Therefore, at the $r^{th}$ communication round, the output of the model of the student $k$ converges to:
\clearpage
\begin{strip}
\begin{align}
\mathbf{F}^{}_k(r) 
&=\frac{1}{\alpha+\lambda}\left[\alpha \mathbf{y}+\frac{\lambda}{K-1} \sum_{k^{\prime}\neq k}^{K} \mathbf{F}_{k^{\prime}}^{}(r-1)\right] \label{eq:fr} \\
&=\frac{1}{\alpha+\lambda} \left[ \alpha \mathbf{y} + \frac{\lambda}{K-1} \sum_{k^{\prime}\neq k}^{K} \left( \frac{1}{\alpha+\lambda}\left[\alpha \mathbf{y}+\frac{\lambda}{K-1} \sum_{k^{"}\neq k^{\prime}}^{K} \mathbf{F}_{k^{"}}^{}(r-2)\right] \right) \right]  \\
&=\frac{1}{\alpha+\lambda} \left[ \alpha \mathbf{y} + \frac{\lambda}{(K-1)(\alpha+\lambda)} \left(\alpha(K-1)\mathbf{y} + \frac{\lambda}{K-1} \sum_{k^{\prime}\neq k}^{K} \sum_{k^{"}\neq k^{\prime}}^{K} \mathbf{F}_{k^{"}}^{}(r-2)\right) \right]   \\
&=\frac{1}{\alpha+\lambda} \left[ \alpha \mathbf{y} + \frac{\lambda}{(K-1)(\alpha+\lambda)} \left(\alpha(K-1)\mathbf{y} + \lambda \mathbf{F}^{}_k(r-2) + \frac{\lambda(K-2)}{K-1} \sum_{k^{\prime}\neq k}^{K}  \mathbf{F}_{k^{\prime}}^{}(r-2)\right) \right] \label{eq:frr} 
\end{align}
\end{strip}
On the other hand, \eqref{eq:fr} can be rewritten as: 
\begin{equation}
\label{eq:frnew}
    (\alpha+\lambda)\mathbf{F}^{}_k(r) - \alpha \mathbf{y} =\frac{\lambda}{K-1} \sum_{k^{\prime}\neq k}^{K} \mathbf{F}_{k^{\prime}}^{}(r-1).
\end{equation}
Then, let
\begin{equation}
\label{eq:frrec}
 u_{r+1} =\frac{\lambda}{K-1} \sum_{k^{\prime}\neq k}^{K} \mathbf{F}_{k^{\prime}}^{}(r)=(\alpha+\lambda)\mathbf{F}^{}_k(r+1) - \alpha \mathbf{y}.
\end{equation}
By introducing $u_{r+1}$ in \eqref{eq:frr}, we can construct a linear non-homogeneous recurrence relation given by:
\begin{equation}
    \begin{aligned}
\mathbf{u}_{r+1}&=\frac{(K-2) \lambda}{(K-1)(\alpha+\lambda)} \mathbf{u}_{r}+\frac{\lambda^{2}}{(K-1)(\alpha+\lambda)^{2}} \mathbf{u}_{r-1}\\
&+\frac{\lambda^{2} \alpha \mathbf{y}}{(K-1)(\alpha+\lambda)^{2}}+\frac{\lambda \alpha \mathbf{y}}{(\alpha +\lambda)}.
\end{aligned}
\end{equation}
The resolution of such recurrence relation as presented in \cite{rosen_2012},  provides a close-form solution defined as:
\begin{equation}
\mathbf{u}_{r}= \lambda \mathbf{y} + \gamma\left(\frac{\lambda}{\alpha+\lambda}\right)^{r}+\rho\left(-\frac{\lambda}{(K-1)(\alpha+\lambda)}\right)^{r}
\end{equation}
with $\gamma=\frac{\lambda}{K} \sum_{k=1}^{K} \mathbf{F}^{k}(0)-\lambda \mathbf{y}$ and \\ $\rho=\frac{\lambda}{K(K-1)} \sum_{k^{\prime}\neq k}^{K} \mathbf{F}^{k^{\prime}}(0)-\frac{\lambda}{K} \mathbf{F}^{k}(0)$. Combining \eqref{eq:frnew} and \eqref{eq:frrec}, we obtain the following:
\begin{equation}
    \mathbf{F}^{k}(r)=\frac{1}{\alpha+\lambda}(\alpha \mathbf{y}+\mathbf{u}_r).
\end{equation}
It follows that: 
\begin{equation}
    \lim _{r \rightarrow \infty} \mathbf{F}^{k}(r)=\frac{1}{\alpha+\lambda}(\alpha \mathbf{y}+ \lim _{r \rightarrow \infty}\mathbf{u}_r) =\mathbf{y},
\end{equation}
due to the fact that $\lim _{r \rightarrow \infty}\mathbf{u}_r = \lambda \mathbf{y}$, for $K\geq 2$. Indeed for $K \geq 2$, $\left|-\frac{\lambda}{(K-1)(\alpha+\lambda)}\right|<1$ and $\left|\frac{\lambda}{\alpha+\lambda}\right|<1$.
\end{proof}
Extensive numerical evaluation using  simulation as well as  experimental data in Section~\ref{sec:performanceevaluation} confirm the convergence property of our FD system.
\subsection{Communication System}
\label{sec:com}
\subsubsection{Energy consumption}
In a FL system as well as a FD system, regarding the IoT network resources limitations, the overall system efficiency heavily relies on the communication system since the energy consumption $\mathcal{E}$ of every single device participating in the training process is expressed as:
\begin{equation}
   \mathcal{E}=\left(\mathcal{E}_{\mathrm{C}}+\mathcal{E}_{\mathrm{T}}\right) \times \mathcal{N}_{\mathrm{T}},
   \label{eq:energy}
\end{equation}
where $\mathcal{E}_C$ is the computation energy of each device during the local training at each training step or communication round. $\mathcal{E}_T$, the transmission energy, is the quantity of energy used by each device to transmit its parameters to the server at each communication round, and $\mathcal{N}_{\mathrm{T}}$ is the total number of training steps till convergence. 

From ~\eqref{eq:energy}, we can see that the total energy consumption depends on three major components namely: $(i)$ the ML model parameters size,
$(ii)$ the local computation and $(iii)$ the number of communication rounds.

Yet our proposed FD framework by design solves the problem $(i)$. In fact, instead of exchanging model parameters, FD devices only transmit their predictions to the server which greatly reduces the communication cost. However, this reduction causes a slight accentuating of the problem $(ii)$ because of the regularization term in the devices updated loss functions. Lastly, the problem $(iii)$ is data dependant.
Therefore, for this work, assuming the same number of communication rounds for both \ac{FL} and \ac{FD}, we only focus on the transmission energy $\mathcal{E}_T$ which is the dominant part of the system overall energy consumption.

To clearly illustrate the communication behaviour of our framework, we consider a simple multiple access communication system. 
In fact, on the uplink, FD clients share a Gaussian multiple-access channel whose equation is given by:
\begin{equation}
\mathbf{y}_r=\sum_{c=1}^{C}  \mathbf{h}_r^{c}\mathbf{x}_r^{c}+\mathbf{z}_r,
\end{equation}
where $\mathbf{x}_r^{c}$ is the signal to be transmitted by client $c$ at the $r^{th}$ communication round, $\mathbf{h}_r^{c}$ the corresponding channel response and $\mathbf{z}_r$ an additive \ac{I.I.D} Gaussian noise with  $\mathbf{z}_r \sim \mathcal{N}(0,1)$.
For the sake of simplicity, the downlink communication is assumed to be noiseless so that we can focus on the more challenging shared uplink.
Thus, for FD, at each iteration, each client transmits its local average estimations per segment over a  wireless  shared  uplink channel  with  the access point connected to the server. Assuming that the target variables have dimension $N_o$ each divided into $\mathcal{S}$ segments, 
the output vector to be transmitted is given by:
\begin{equation}
    \mathbf{V}= \begin{bmatrix}
\mathbf{\bar{f}}_1^c \\ 
\mathbf{\bar{f}}_2^c \\ 
 \vdots \\ 
\mathbf{\bar{f}}_s^c \\ 
\vdots \\ 
\mathbf{\bar{f}_\mathcal{S}}^c 
\end{bmatrix}
 \ \ \ \text{with}\ \ \   \mathbf{\bar{f}}_s^c = \begin{bmatrix}
 \mathbf{\bar{f}}_s^c(1)\\ 
 \mathbf{\bar{f}}_s^c(2) \\ 
\vdots\\ 
 \mathbf{\bar{f}}_s^c(L)
\end{bmatrix}^T.
\end{equation}
Thus, unlike FL where each device transmits $ W \times R$ bits to the server for each communication round, in FD the total number of transmitted bits per device per communication round is given by:
\begin{equation}
    \mathcal{N}_{bits} = \mathcal{S} \times N_o \times R,
    \label{equ:bits}
\end{equation}
where $R$ is the bits resolution. The total number of parameters of our considered neural network is defined as:
\begin{equation}
    W = (N_i \times N_h)+N_h+ (N_h \times N_o) + N_o. 
    \label{equ:weights}
\end{equation}
So obviously, it follows that $\mathcal{S} \times N_o \ll W$ which makes our model much more communication-efficient than FL. 
This lightweightness of our system overcomes the bandwidth and energy limitations since each device involved in the training process is subject to the power constraint $\mathrm{E}\left[\left\|\mathrm{x}_r^{c}\right\|_{2}^{2}\right] / T \leq \mathcal{P}$ \cite{network},
where $T$ is the number of channels and $\mathcal{P}$ is the maximum transmit power.
Indeed, considering a conventional digital implementation as in \cite{com}, the channel's uplink capacity is equally shared by the devices leading to bandwidth constraints limiting the number of bits $B$ that can be transmitted by each device. In fact, considering the Shannon’s capacity \cite{shannon}, the maximum number of bits per transmission for each device is given by: 
\begin{equation}
    {B}_{max}^{r,c}=\frac{T}{ C} \log _{2}(1+C \left | \mathbf{h}_r^c \right |^2 \mathcal{P}), \text{ so that  } {B} \leq {B}_{max}^{r,c}.
\end{equation}
To sum up, this constraint does not affect too much the FD system since the condition $\mathcal{N}_{bits} < B_{max}^{r,c}$ is very likely to always hold as demonstrated by the results in Section~\ref{sec:performanceevaluation}.

\section{Performance Evaluation }
\label{sec:performanceevaluation}

\subsection{Indoor Localization with Federated Distillation}
\subsubsection{System model and data collection}
Our goal is to construct a DNN based indoor positioning system (IPS) to predict the location of an IoT device given its RSSI measurements in that location. Therefore, the dataset is an ensemble of measurements of RSSI from different \acp{AP} at different locations and at different time to capture the channel variations. To do so, we consider an indoor IoT network in an area of interest of size $l\times w\  m^2$ containing $M$ \acp{AP} randomly distributed over the area. For the data crowdsourcing, we consider a pool of $C$ FD clients. Each client will collect RSSI data from the environment to train a local model. The environment is characterized by its path loss exponent $\beta$ and its sigma-shadowing $\sigma$ representing the channel attenuation.\\
In this localization system we exploit the RSSI values gathered from different \acp{AP} in the network. Consequently, for the data collection we consider $N$ \acp{RP} such that at each \ac{RP} we capture $M$ RSSI values from the $M$ \acp{AP}, taking into account the out-of-range \acp{AP}. At each captured RSSI, the corresponding RP is added as a label.  Moreover, at each \ac{RP}, the operation is repeated for $k$ time intervals in order to capture the variations and different impairments experienced by the wireless channel. \\
\subsubsection{Simulated RSSI data}
The RSSI measurements are strongly time-space varying and depend on the path loss model for the corresponding environment.
In this work, similar to work in \cite{Njima2021}, we use the log-distance path loss model formulated as follows: 
\begin{equation}
    PL=PL_{0}+10\beta \log _{10}{\frac {d}{d_{0}}}+X_{\sigma},
\end{equation}
where $\beta$ is the path loss exponent, $d_0$ is the reference distance, $X_\sigma$ represents the log-normal shadowing with standard deviation $\sigma$ in dB. $PL_0$ is the path loss at the reference distance $d_0$ and given by the free space propagation model as:
\begin{equation}
    {PL_0} =10\log _{10}\left(\left({4\pi d_0\cdot \frac {f}{c}}\right)^{2}\right),
\end{equation}
with $f$ the frequency of the signal and $c$ the speed of light.


Then the RSSI is computed as follows:
\begin{equation}
{\begin{aligned}
RSSI &= P_{Rx} = P_{Tx} - PL. \\
\end{aligned}}
\end{equation}
Consequently, the RSSI value captured from the $m^{th}$ AP at the $n^{th}$ RP at time $t$ can be expressed as:
\begin{equation}
    \begin{aligned}
    RSSI_{n,t}^m &= P_{Tx} - \left[20\log _{10}\left(\frac{4\pi d_0}{c}\right)+20\log _{10}\left(f\right)+ \right. \\
    &\left. 10\beta \log _{10}\left(\frac {d_n^m}{d_{0}}\right)+X_{\sigma,{n,m}}^t\right].\\
    \end{aligned}
    \label{equ:rssi}
\end{equation}

\subsection{Simulation Results}
\begin{figure*}[!t]
    \centering
    \begin{subfigure}[b]{0.33\textwidth}
        \centering
        \includegraphics[width=\textwidth,height=0.2\textheight]{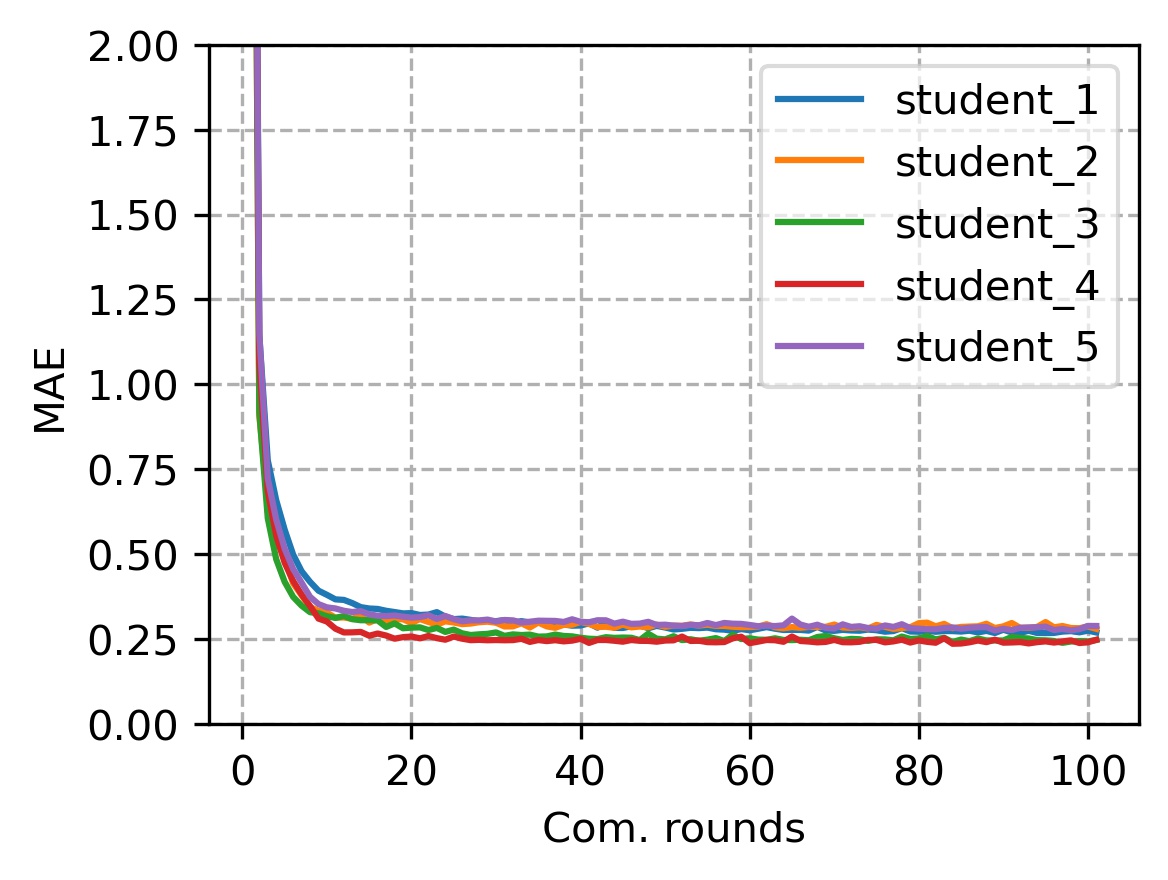}
        \caption{Convergence of student models}
        \label{fig:sim-conv}
    \end{subfigure}%
    \hfill
    \begin{subfigure}[b]{0.33\textwidth}
        \centering
        \includegraphics[width=\textwidth,height=0.2\textheight]{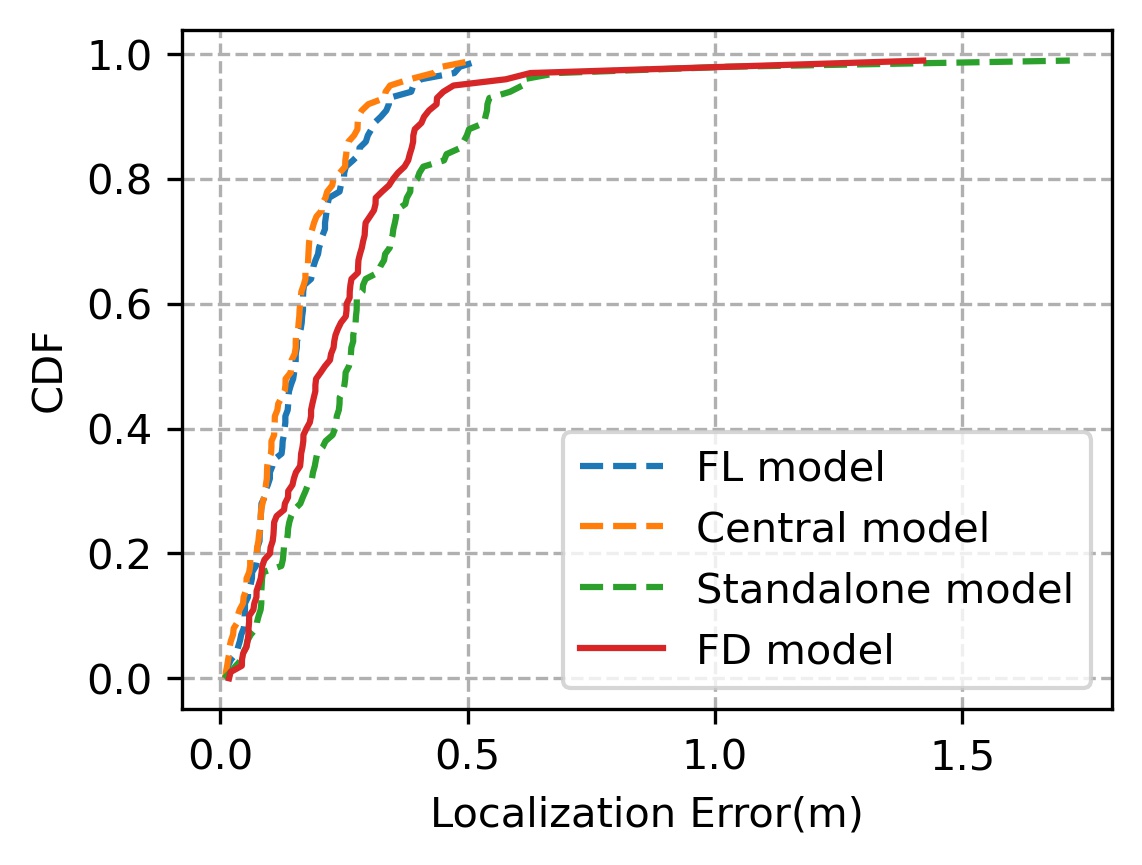}
        \caption{Performance comparison  }
        \label{fig:sim-cdf}
    \end{subfigure}
    \hfill
    \begin{subfigure}[b]{0.33\textwidth}
        \centering
        \includegraphics[width=\textwidth,height=0.2\textheight]{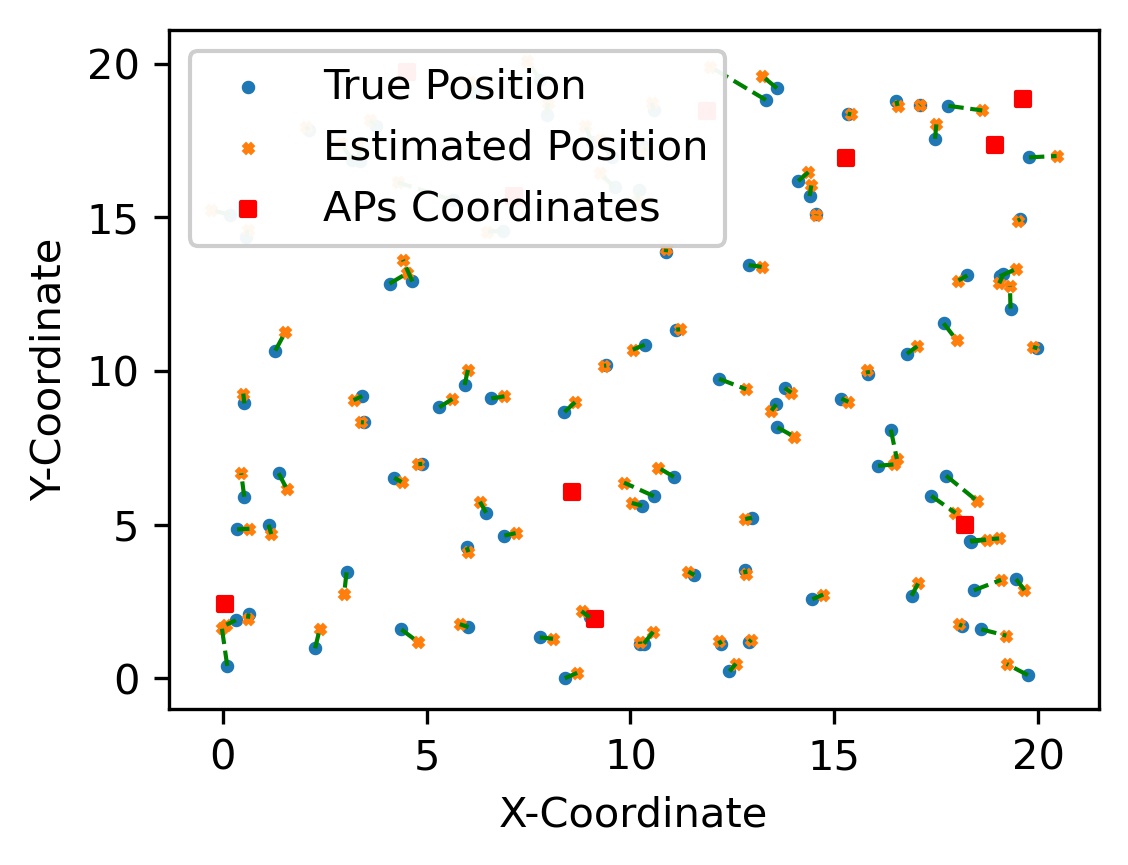}
        \caption{Estimation errors of FD model}
        \label{fig:sim-net}
    \end{subfigure}
    \caption{Regression based FD framework evaluation on simulation data.}
    \label{fig:sim-fdr}
\end{figure*}
\subsubsection{Set-up}
For the simulation, we consider a Wi-Fi powered IoT network of $l\times w=20\times20m^2$ with $M=10$ APs where the data are collected from $N=100$ RPs. The environment variables $\beta$ and $\sigma$ are set to $3.23$ and $2$ respectively according to the experimental measurements conducted in \cite{Njima2020}. 
The number of repetitions is set to $T=10$.
In order to make the dataset reproducible, the RSSI values are generated using \eqref{equ:rssi} with a random seed set to 200.\\
The objective being the training of DNN localization models using federated distillation, we set the number of clients (workers or students) $C$ to $5$.
For the DNN models, we consider a \ac{MLP} characterized by the configuration given in TABLE \ref{tab:dnn-config}.
\begin{table}[!ht]
\caption{Parameters of the \ac{MLP} model }
\centering
    \begin{tabular}{ |m{4em}|m{12em}|m{4em}| }
   \hline
 Parameter & Description & Value  \\
 \hline
 Optimizer & DNN Model Optimizer & Adam \cite{adam}  \\
 \hline
 $\mu$ & Learning rate & 0.0001    \\
\hline
$\beta_1, \beta_2$  & Exponential decay rates & $0.1, 0.99$  \\
\hline
$N_i$ & Input layer units & 10\\  
\hline
$N_h$ & Hidden layer units & $1\times1000$\\  
\hline
$N_o$ & Output layer units & 2\\  
\hline
$\sigma_h(\cdot)$ & Hidden Activation function & ReLu\\  
\hline
$\sigma_o(\cdot)$ & Output Activation function & Linear\\  
\hline
$\mathcal{C}$ & Number of students (clients) & 5\\  
\hline
$b$ & Batch size & 32\\  
\hline
$\mathcal{R}$ & Communication rounds for FL and FD & 100\\  
\hline
$\lambda$ & Regularization term & 0.1\\  
\hline
\end{tabular}
\label{tab:dnn-config}
\end{table}

\subsubsection{Results}
The models are trained by the $C$ students and the FD server using Algorithm \ref{alg:fdr}. The evaluation is done using the mean absolute error (MAE) metric defined by:
\begin{equation}
\begin{aligned}
MAE &=\mathbb{E}\left [  \mathbf{y}-\mathbf{\hat y} \right ]  = \frac{1}{N}\sum_{i=1}^{N}\left \| \mathbf{y}_i - \mathbf{\hat y}_i \right \|_1 \\
 & = \frac{1}{N}\sum_{i=1}^{N} \frac{\left |x_i-\hat x_i\right | + \left |y_i-\hat y_i\right |}{2},
\end{aligned}
\end{equation}
where $\mathbf{y}_i = [x_i\  y_i]^T$ and $\mathbf{\hat y}_i = [\hat x_i \ \hat y_i]^T$ are respectively the true and estimated locations corresponding to the $i^{th}$ \ac{RSSI} recording in the validation dataset.
The results are presented in Fig.~\ref{fig:sim-fdr}.
In fact, we can see from Fig.~\ref{fig:sim-conv} that all the students' models converge to the same MAE value, meaning that the federated students are effectively learning from each other in order to consolidate their models.
To analyse the performance of our framework, we consider different learning scenarios namely federated learning (FL), centralized learning (CL) which is the traditional training approach, and standalone learning (SL) where the student train independently its model. As such, in Fig. \ref{fig:sim-cdf}, a single data point per reference position is chosen to feed the trained models and the predictions are compared to the ground-truth labels. We can see that the FD model improves the localization accuracy compared to the SL since the students share their knowledge with each other. Nonetheless, it remains less accurate than FL and CL due to the nature of its operating mode. It is important to note that this accuracy decreases is the price of a huge communication gain which is the principal goal of this work. 
Indeed, in terms of communication complexity as shown in TABLE~\ref{tab:rmse}, with a bits resolution $R=32$, the number of segments $\mathcal{S} =10$ and the output dimension $N_o = 2$ using $C=5$ students, the \ac{FD} model is far better than the \ac{FL}. In fact, with \ac{FD}, only $5\times640 = 3200$ bits \eqref{equ:bits}
are transmitted at each round in comparison to $5\times 416064 = 2080320$ bits \eqref{equ:weights} for \ac{FL}, leading to a FD-to-FL ratio of $0.15\%$. Consequently, \ac{FD} can save up to $99.85\%$ of the transmission energy $\mathcal{E}_T$ used in \ac{FL} while remaining only $1.6\times$ less accurate than \ac{FL}.

Finally, we can see from Fig. \ref{fig:sim-net} how accurate is our localization framework by observing the estimation errors plotted therein. The network overall \ac{RMSE} is computed for each model and the results are presented in TABLE \ref{tab:rmse}.
These results attest that our framework is suitable for localization in indoor IoT network since it presents relatively good accuracy regarding the geometry of the network and preserves network resources.

\subsection{Experimental validation and benchmark}
\begin{figure*}[!t]
     \centering
     \begin{subfigure}[b]{0.325\textwidth}
         \centering
         \includegraphics[width=\textwidth,height=0.2\textheight]{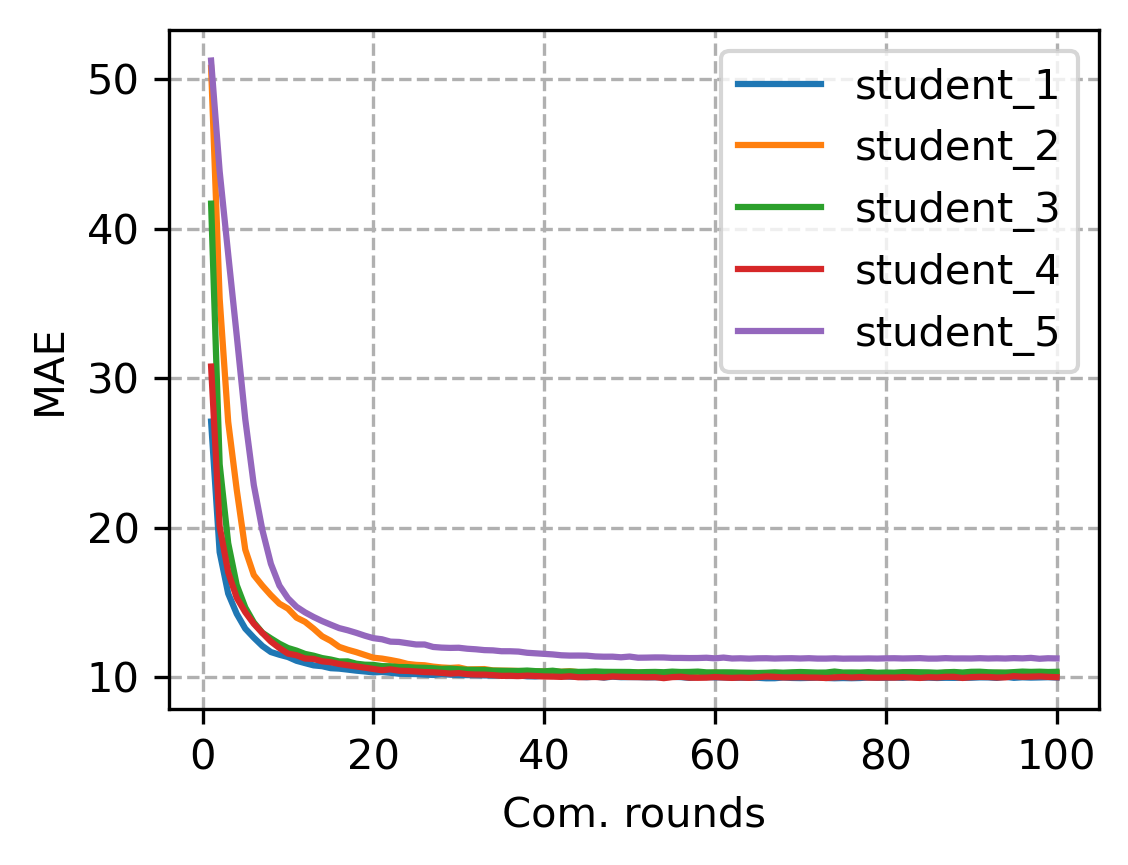}
         \caption{Convergence of \ac{FD} models}
         \label{fig:exp-uji-conv}
     \end{subfigure}
     \hfill
     \begin{subfigure}[b]{0.325\textwidth}
         \centering
         \includegraphics[width=\textwidth,height=0.2\textheight]{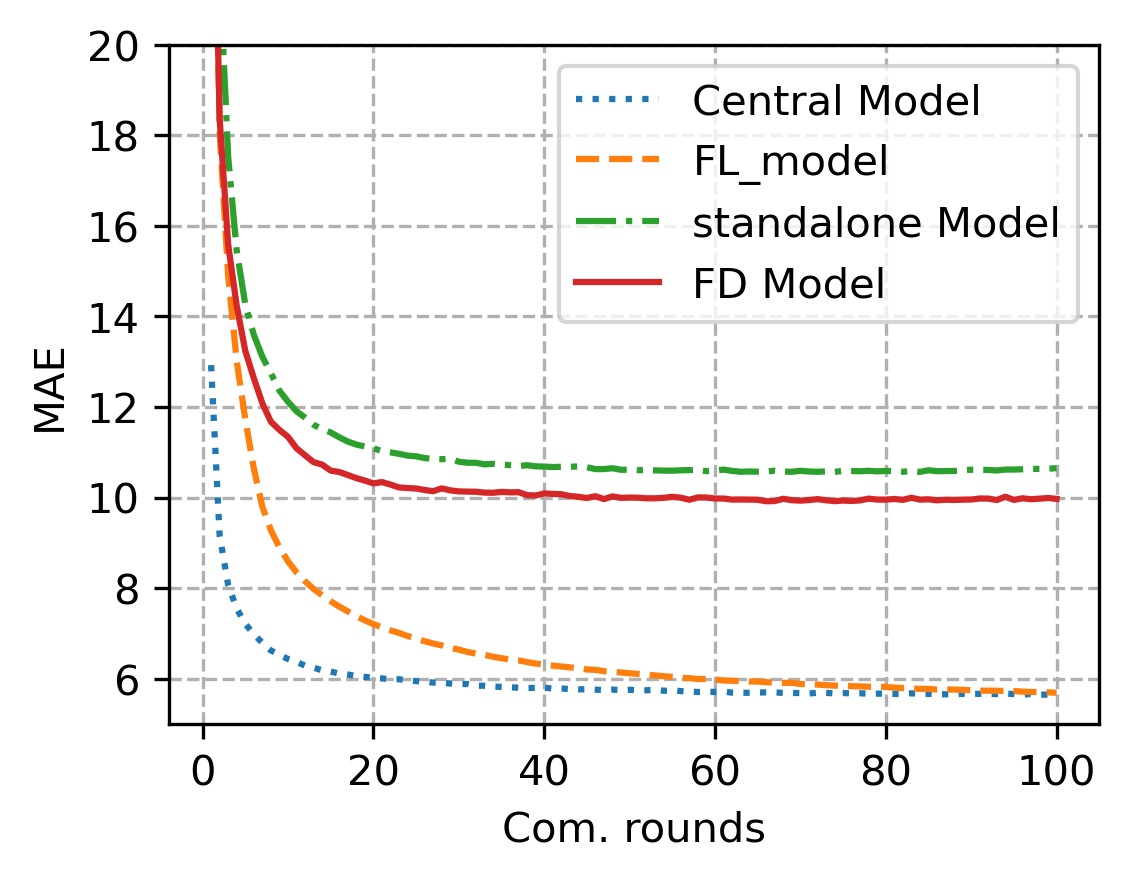}
         \caption{Models comparison}
         \label{fig:exp-uji-bchmk}
     \end{subfigure}
     \hfill
     \begin{subfigure}[b]{0.325\textwidth}
         \centering
         \includegraphics[width=\textwidth,height=0.2\textheight]{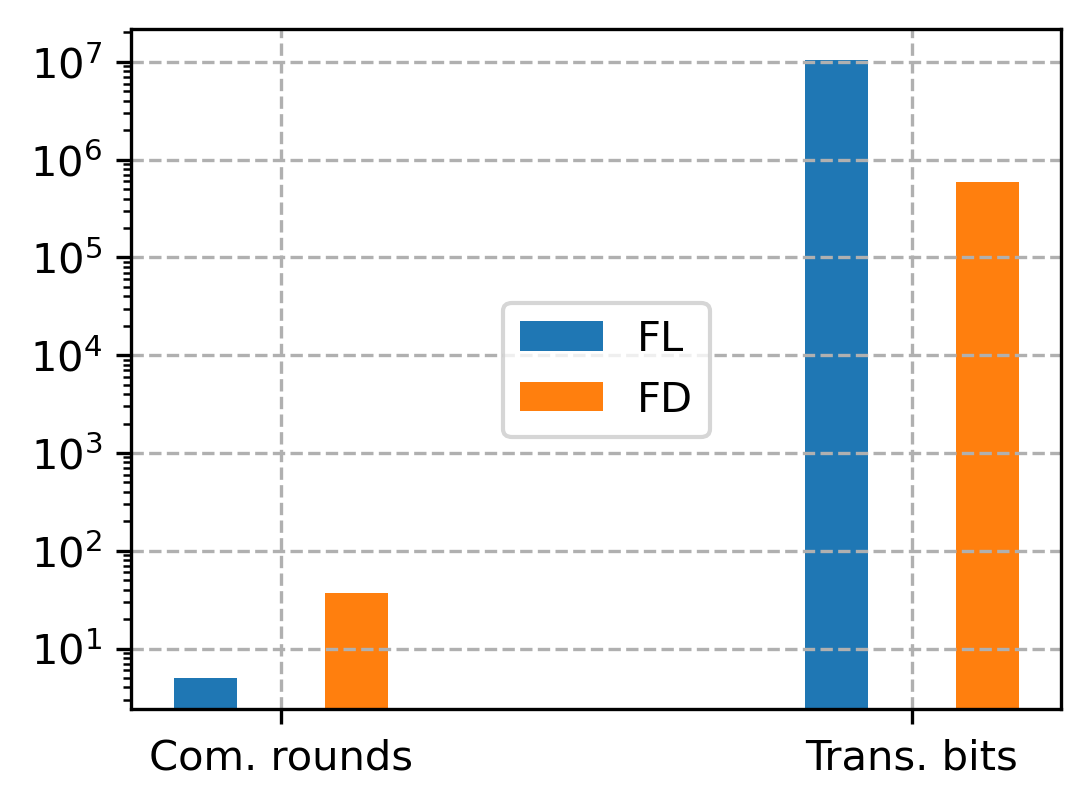}
         \caption{Comparison for MAE = 10m}
         \label{fig:exp-uji-smae}
     \end{subfigure}
        \caption{Regression based FD framework evaluation on UJIIndoorLoc experimental data.}
        \label{fig:exp-uji}
\end{figure*}
In this part, we implement our framework using experimental localization data as well as other regression tasks data.
First of all, we perform indoor localization using UJIIndoorLoc WiFi fingerprinting dataset \cite{ujiindoorloc} collected on a campus area of approximately $390m\times270m$ containing $21048$ entries. The results are presented in Fig.~\ref{fig:exp-uji}. Fig.~\ref{fig:exp-uji-conv} shows the convergence of all the MAE of all students to the same MAE value which confirms the convergence of the proposed FD algorithm. In Fig.\ref{fig:exp-uji-bchmk}, our proposed FD model  is benchmarked with the newly introduced FL model in \cite{FLLcrowd2020} in the context of indoor localization using the same dataset in \cite{ujiindoorloc}. Note that we used the same \ac{DNN} model presented in \cite{FLLcrowd2020} on all the datasets. In order to highlight the relevance of our FD framework, we keep showing the results of the standalone model (independent learning) as well as the centralized system whenever possible. It comes out that the FD learning slightly improves the estimation accuracy compared to its standalone (baseline) model, while remaining less accurate than the FL approach. 

The benefit of FD over FL in terms of energy efficiency is seen in Fig.\ref{fig:exp-uji-smae} where the number of communication rounds and the transmitted bits of FD are compared with those of FL for a fixed accuracy (MAE = 10m). In fact, even though it takes several rounds, FD transmits fewer bits to achieve the same accuracy as FL. Due to the fact that the computation energy is lower in comparison to the transmission energy, the FD still consumes significantly less energy when compared to the FL. This is a substantial advantage for IoT systems with very limited bandwidths. Therefore, for the same localization accuracy, the proposed FD framework comes with a greater energy efficiency by a reasonable saving of transmission energy.
Simply put, in addition to the slight increase of the accuracy of FD regarding its baseline model, the main advantage of this framework over FL relies in its communication efficiency as shown in TABLE~\ref{tab:bitscomp}. Therein, it can be seen that \ac{FD} can save more than $98\%$ of the transmitted bits per each communication round. Consequently, the transmission energy $\mathcal{E}_T$ is minimized.

Moreover, we perform outdoor localization in an urban \ac{LoRaWAN} using the dataset reported in \cite{urbanlora} and diamond prices prediction using the dataset in \cite{diamonds}. We also present the results of the baseline models used in the original works where the datasets are presented. The results are summarized in TABLE~\ref{tab:bitscomp} for the communication complexity and TABLE \ref{tab:rmsecomp} for the accuracy. 
Note that for fair comparison between different systems, the same number of epochs are used for standalone training as well as for centralized training considering the number of local epochs per round and the number of communication rounds in the federated settings. It worth noting that for UJIIndoorLoc and LoRa datasets, the target variables have dimension 2 (2-D localization) while diamond dataset's is  1-D (the diamond's price). The results assert that our framework can be used for any localization system in wireless networks and by extent, for any regression task since it has also shown great performance with the diamond pricing task which is a typical regression problem.

\begin{table}[!t]
\caption{Network global \ac{RMSE}}
\centering
    \begin{tabularx}{0.4\textwidth} { 
  | >{\raggedright\arraybackslash}X
  | >{\raggedright\arraybackslash}X
  | >{\centering\arraybackslash}X| }
   \hline
 Model & RMSE(m) & bits/round  \\
 \hline
 FD Model & 0.35 & 3200  \\
 \hline
 Standalone Model & 0.38 & -   \\
\hline
FL Model  & 0.21  & 2080320 \\
\hline
Central Model  & 0.20 & - \\  
\hline
\end{tabularx}
\label{tab:rmse}
\end{table}

\begin{table}[!t]
\caption{Comparison of bits/round on different datasets}
\centering
  \begin{tabular}{ |m{4em}|m{6em}|m{4em}|m{4em}| }
   \hline
 Model & UJIIndoorLoc\cite{ujiindoorloc} & Urban LoRa\cite{urbanlora} & Diamonds Pricing\cite{diamonds}   \\
 \hline
 FD Model & 3200 & 3200 & 1600 \\
\hline
FL Model & 1739520 & 267520 & 147360  \\
\hline
FD-to-FL ratio & 0.18\% & 1.19\% & 1.08\%  \\
\hline
\end{tabular}
\label{tab:bitscomp}
\end{table}

\begin{table}[!t]
\caption{Comparison of RMSE (m) on different datasets}
\centering
  \begin{tabular}{ |m{4em}|m{6em}|m{4em}|m{4em}| }
   \hline
 Model & UJIIndoorLoc\cite{ujiindoorloc} (MAE) & Urban LoRa\cite{urbanlora} & Diamonds Pricing\cite{diamonds}   \\
 \hline
 FD Model & 19.75 (9.91) & 483.42 & 178.46 \\
 \hline
 Standalone Model &20.67 (10.80) & 484.48 & 206.76  \\
\hline
FL Model & 11.02 (5.76) & 438.44 &  46.99  \\
\hline
Central Model  & 10.81 (5.66) & 440.35 & 40.51 \\  
\hline
Baseline Model & 7.98 (-)  & 398.40 & 462.62\\  
\hline
\end{tabular}
\label{tab:rmsecomp}
\end{table}


\subsection{Scaling Capability}
\vspace{-0.01in}
In this part, we discuss the scalability potential of our framework by varying the number of students involved in the learning process. The accuracy in terms of RMSE and the number of transmitted bits $\mathcal{N}_{bits}$ are highlighted, and the results are presented in Fig. \ref{fig:scalability}.\\
Fig. \ref{fig:rmse}  shows that as the number of students increases, the system accuracy keeps improving for both FL and FD with more effect on the FL system. This observation is straightforward since in FD, students contribute to the learning by only sharing their output in contrast to FL where all the models' parameters are shared in the network. On the other hand, Fig. \ref{fig:bits} presents the communication loads in terms of the total number of transmitted bits by all students to the server, showing the prominence of using FD. We observe that this results confirms \eqref{equ:bits} and effectively the FD is more bandwidth conservative than FL. The FD/FL ratios of these metrics are shown in Fig. \ref{fig:ratio} where we can see that in FD we transmit only $1.46\%$ of the number of bits transmitted in FL, making FD $\sim 98,54\%$ more  efficient than FL in terms of communication while remaining only $\sim 1\times \text{ to } 2.7 \times$ less accurate.


\begin{figure*}[!ht]
    \centering
    \begin{subfigure}[b]{0.33\textwidth}
        \centering
        \includegraphics[width=\textwidth,height=0.2\textheight]{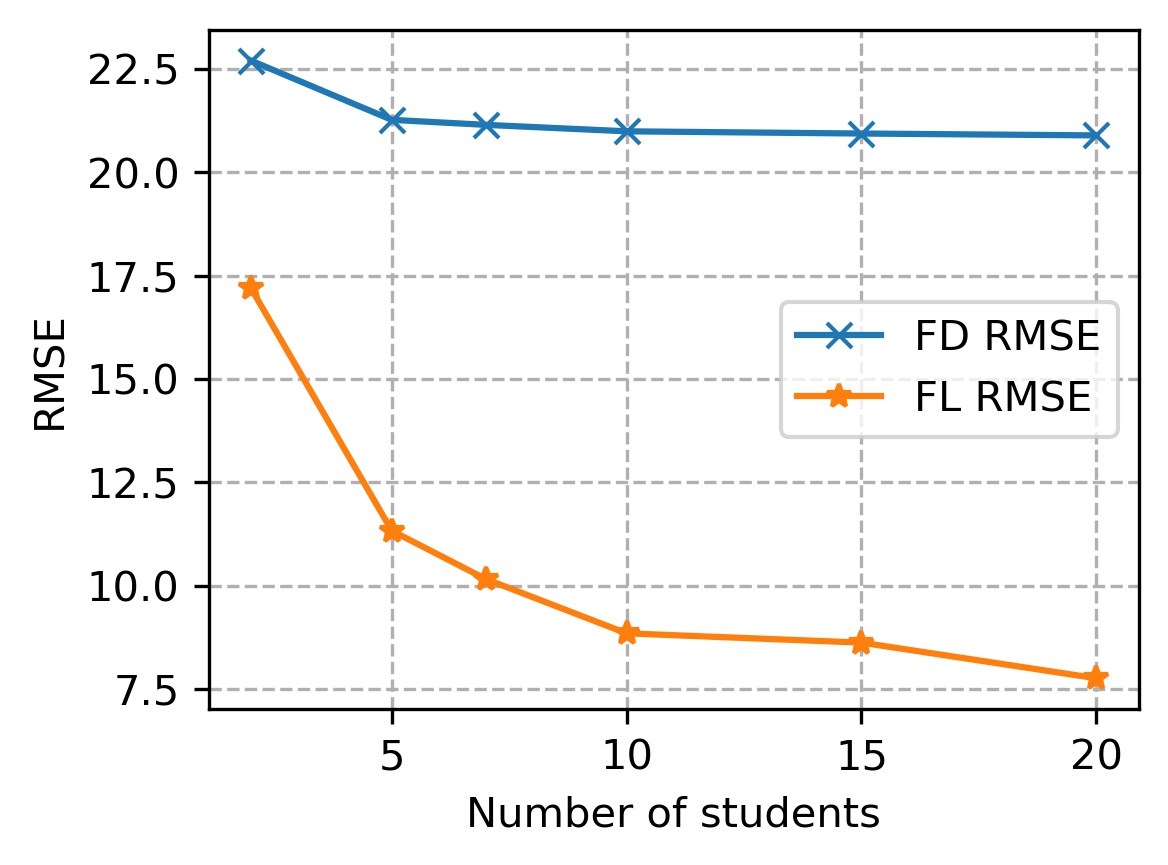}
        \caption{Estimation error}
        \label{fig:rmse}
    \end{subfigure}%
    \hfill
    \begin{subfigure}[b]{0.33\textwidth}
        \centering
        \includegraphics[width=\textwidth,height=0.2\textheight]{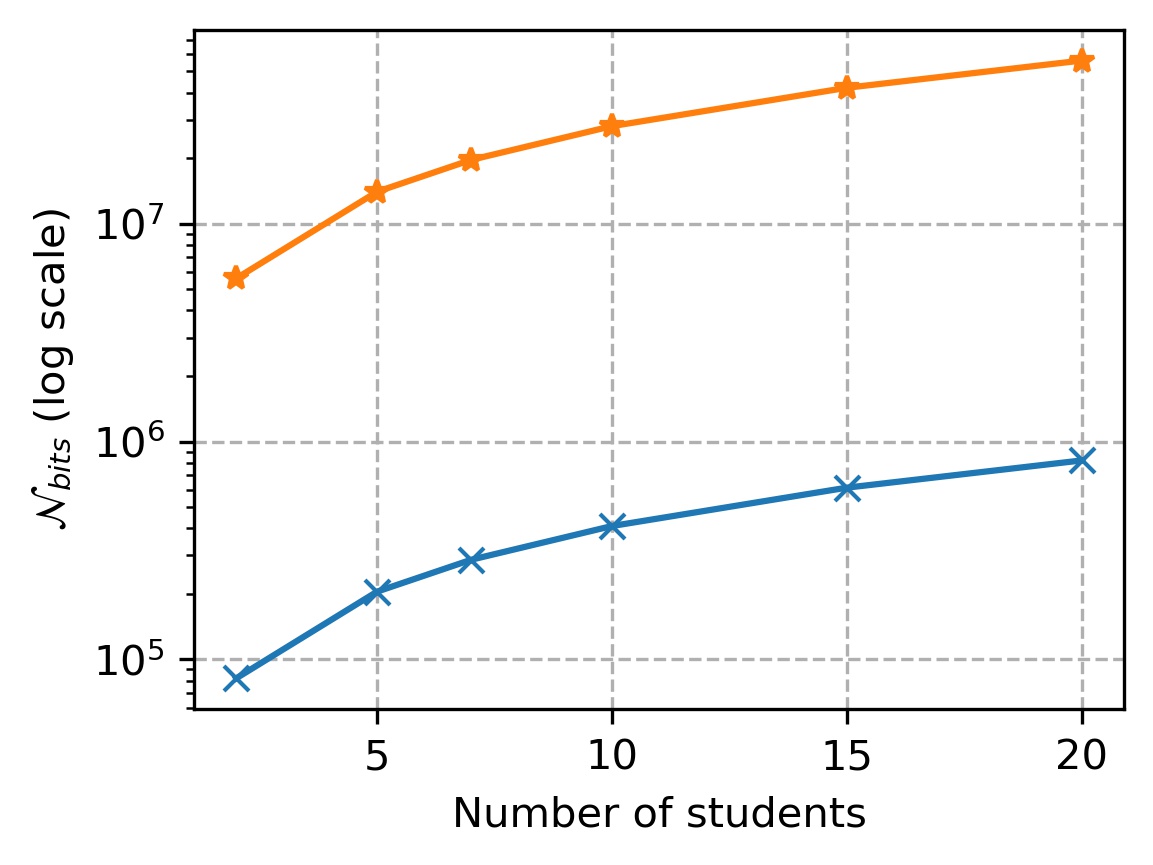}
        \caption{Transmitted bits per round }
        \label{fig:bits}
    \end{subfigure}
    \hfill
    \begin{subfigure}[b]{0.33\textwidth}
        \centering
        \includegraphics[width=\textwidth,height=0.2\textheight]{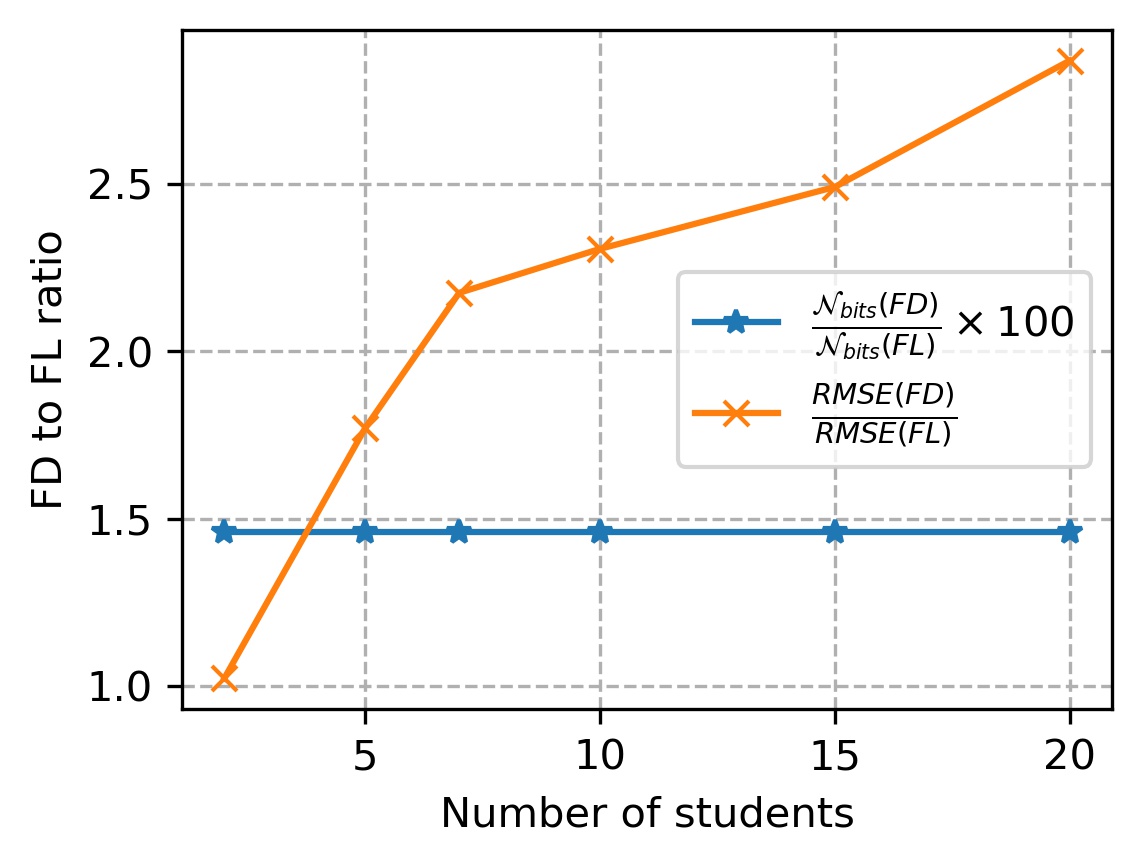}
        \caption{FD to FL metrics ratio}
        \label{fig:ratio}
    \end{subfigure}
    \caption{Regression based FD scalabity analysis, UJIIndoorLoc dataset\cite{ujiindoorloc}}
    \label{fig:scalability}
\end{figure*}

\begin{figure*}[!ht]
    \centering
     
    \begin{subfigure}[b]{0.33\textwidth}
        \centering
        \includegraphics[width=\textwidth,height=0.2\textheight]{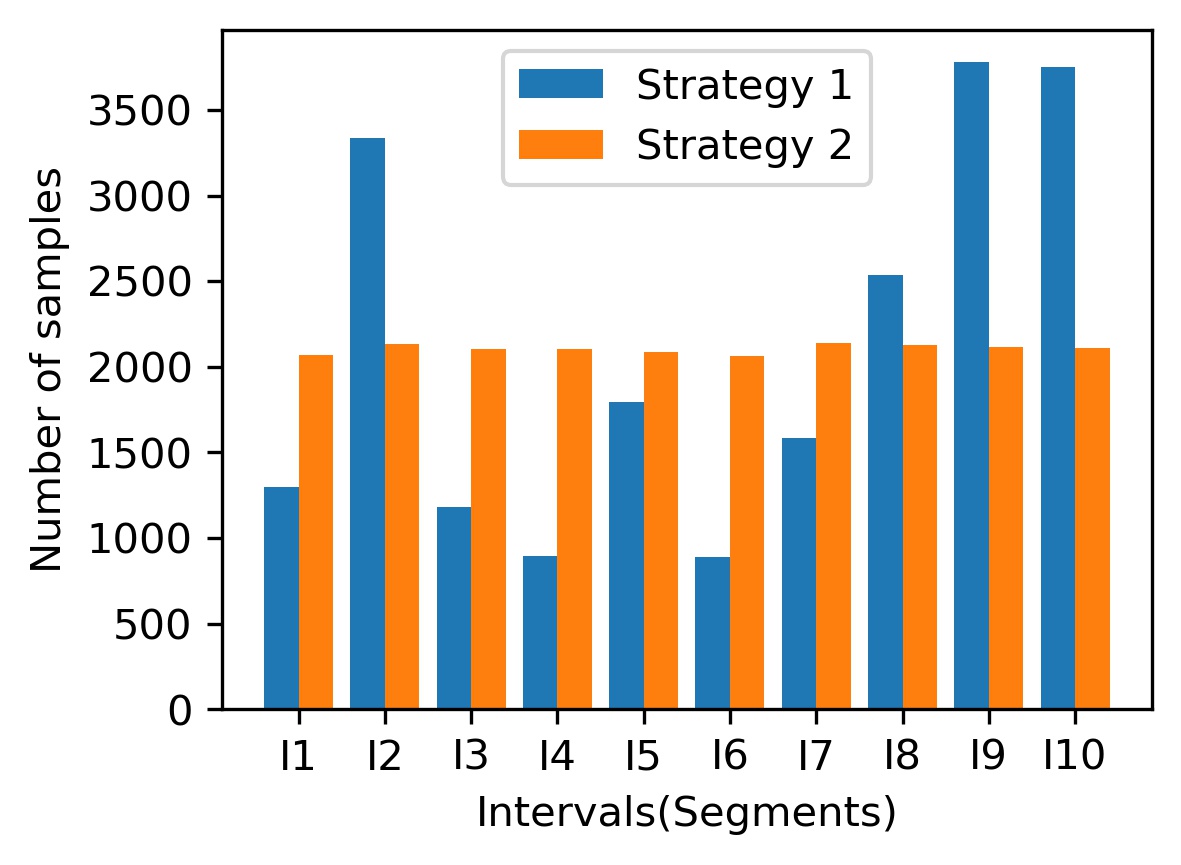}
        \caption{Segmentation strategies}
        \label{fig:strategies}
    \end{subfigure}%
    \hfill
    \begin{subfigure}[b]{0.33\textwidth}
        \centering
        \includegraphics[width=\textwidth,height=0.2\textheight]{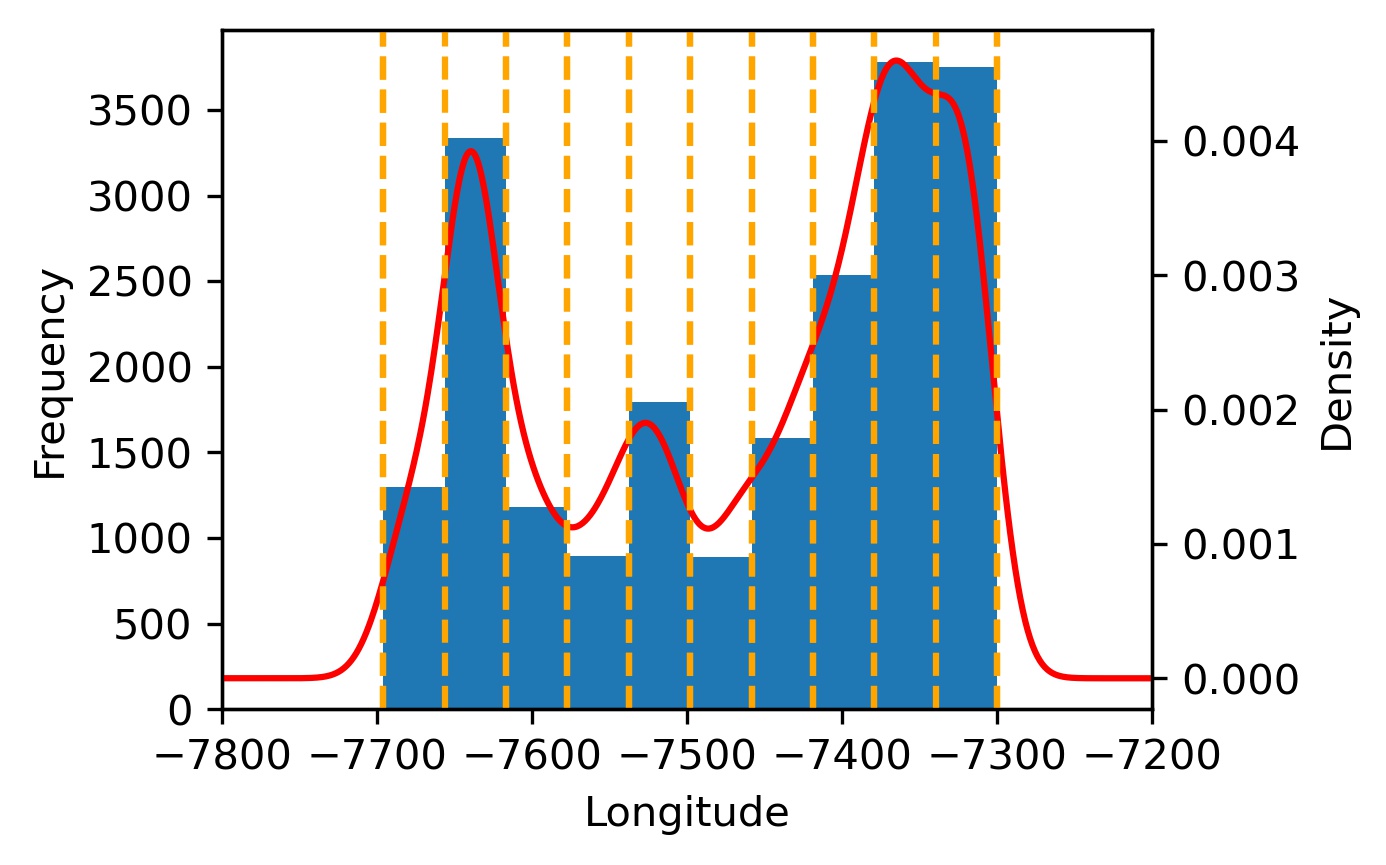}
        \caption{Strategy 1}
        \label{fig:strategy1}
    \end{subfigure}
    \hfill
    \begin{subfigure}[b]{0.33\textwidth}
        \centering
        \includegraphics[width=\textwidth,height=0.2\textheight]{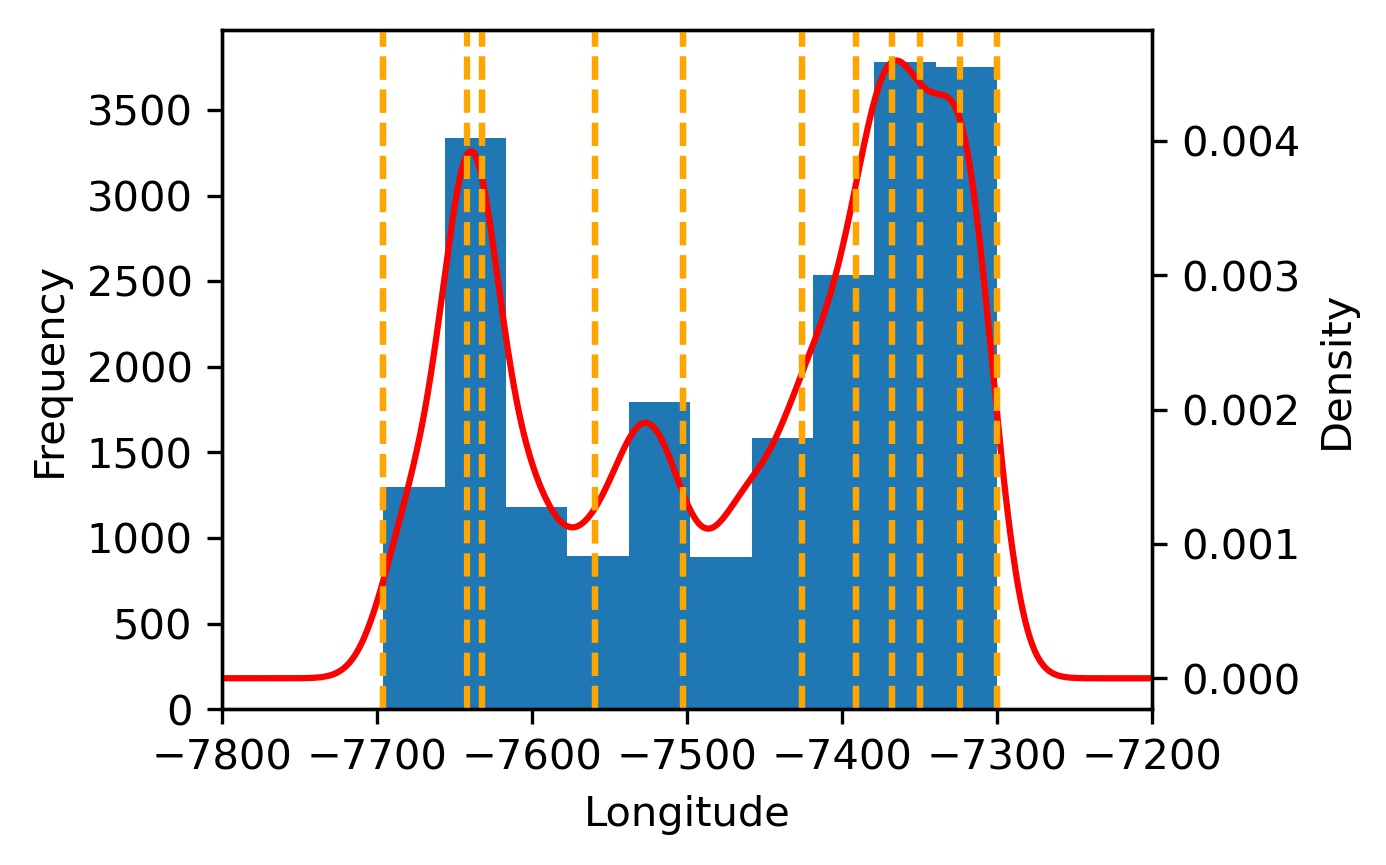}
        \caption{Strategy 2}
        \label{fig:strategy2}
    \end{subfigure}
    
    \caption{Impact of segmentation on data distribution, UJIIndoorLoc dataset\cite{ujiindoorloc}}
    \label{fig:seg-distribution}
\end{figure*}


\begin{figure}[!t]
    \centering
    \includegraphics[scale=0.8]{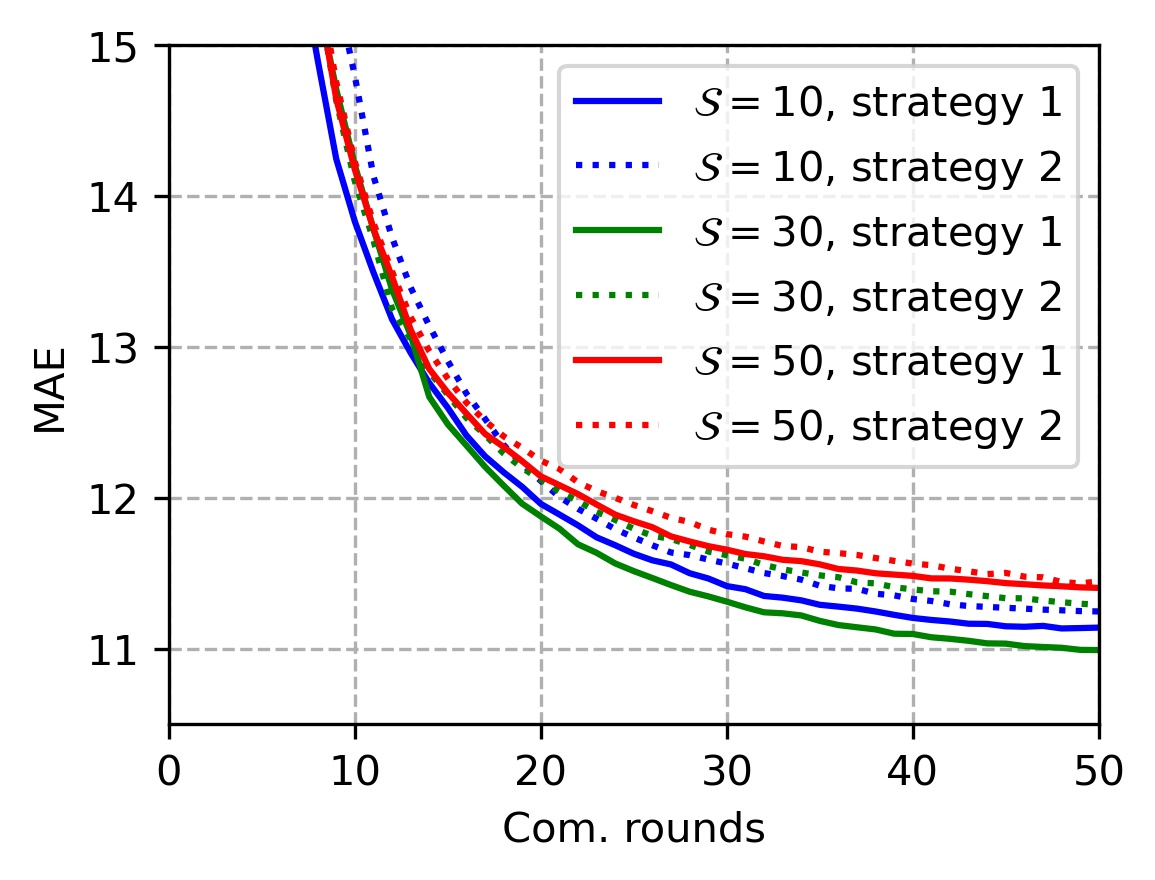}
    \caption{Impact of segmentation on system accuracy}
    \label{fig:seg-accuracy}
\end{figure}

\begin{figure}[!t]
    \centering
    \includegraphics[scale=0.8]{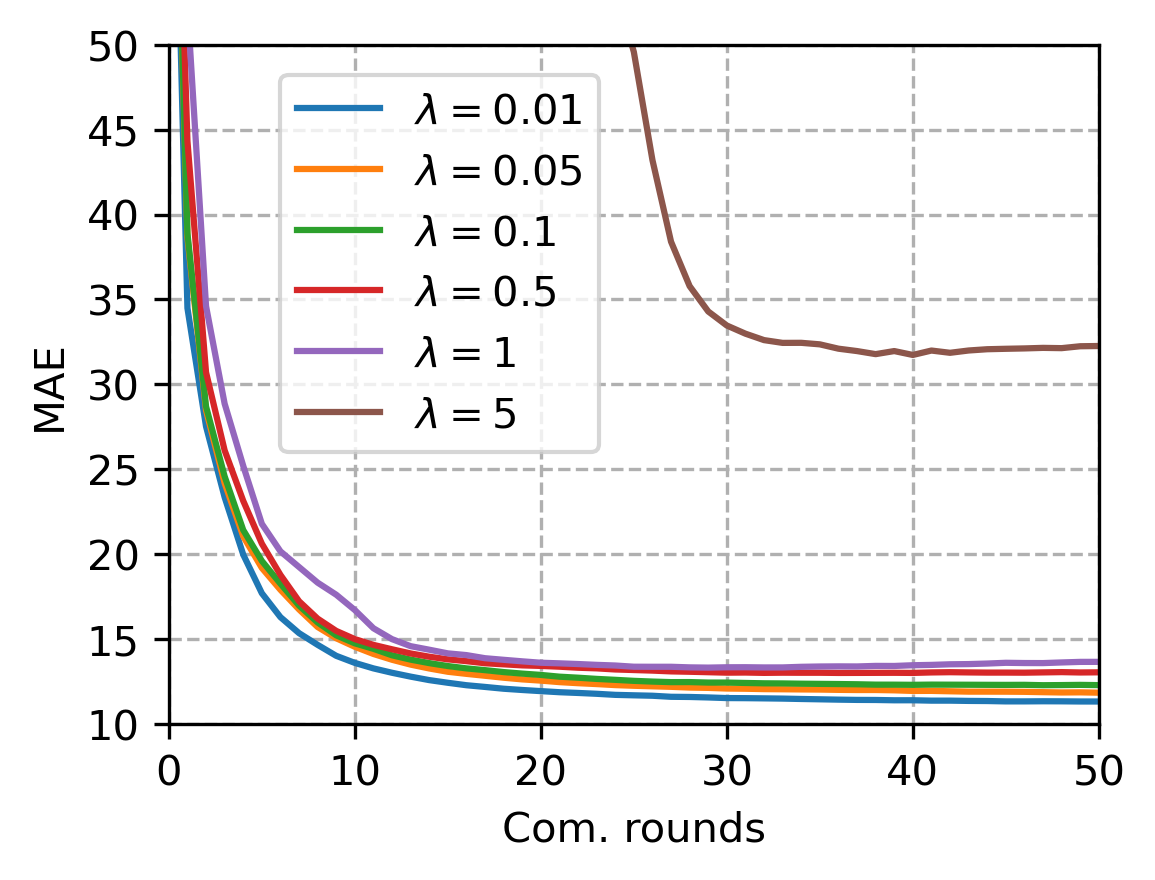}
    \caption{Impact of the distillation regularizer on the FD system accuracy}
    \label{fig:lambda}
\end{figure}

\subsection{The impact of segmentation and data distribution}
\label{sec:segm}
In this part, we analyse the impact of the segmentation on the system performance. Indeed, we consider two segmentation strategies as shown in Fig.~\ref{fig:strategies}.
In the first strategy in Fig. \ref{fig:strategy1}, a \textit{uniform split} which is so far used in this work and was explained in Fig. \ref{fig:seg},  we divide the output target into equal size segments referred to as intervals. So depending on the data distribution, some intervals may be empty or at least very less populated because the target value is far from being uniformly distributed. Consequently, the private datasets of clients will only contain a subset of segments playing the role of labels here. This results in missing label problem which needs to be properly mitigated depending on the dataset.

Alternatively, we design a second segmentation strategy in Fig. \ref{fig:strategy2}, a \textit{density-based split} in which each segment/interval possesses the same number of samples, meaning that their sizes are necessarily different. This will consequently affect the accuracy of the models as shown in Fig. \ref{fig:seg-accuracy}, so that at the end, the choice of a strategy relies on the data distribution.
Moreover, intuitively, increasing the number of segments should improve the models accuracy but actually as shown in Fig. \ref{fig:seg-accuracy}, it can in contrast degrades the model performance if the distribution impairments induced by the segmentation is not properly mitigated. 

     
    

\subsection{Impact of the distillation regularizer}
The FD students learn from their teachers by distilling the teachers' knowledge weighted by a regularizer. This regurlaizer is an hyperparameter of the FD system that plays an import role in how much the students learn from their teachers. Consequently, the choice of the regurlarizer impacts directly the students models accuracy. In Fig~\ref{fig:lambda}, we show the effect of the regularizer $\lambda$ in the FD based indoor localization with the UJIIndoorLoc dataset. We can observe that for this application, the smaller the regularizer the better the accuracy. Thus, a good choice of $\lambda$ for this application has to meet the condition $\lambda<1$.
This result is in accordance with the theoretical analysis conducted in Section~\ref{sec:conv}, where it has been shown that the student prediction error increase with the value of $\lambda$ in the presence of imperfect teacher model.


\section{Conclusion}
\label{sec:conclusion}
In this work, we have presented a novel federated distillation framework meant for regression problems and leverage it to design an indoor IoT localization system. Our solution has shown remarkable results in terms of bandwidth and energy saving compared to \ac{FL} based localization systems. 
The strength of our framework relies in its novelty and its communication efficiency since it
can reduce the communication load by $\sim 98\%$ compared to \ac{FL} framework. In addition to its communication efficiency, our proposed framework presents  a great scalability capability making it suitable for large applications in \ac{IoT} networks. However, despite its great performance, this framework can be further improved by tackling its limitations such as the impairments induced by the segmentation
and devices heterogeneity in terms of computation and communication capabilities.


\bibliographystyle{IEEEtran}
\bibliography{references}




\end{document}